%% file: main.tex
\RequirePackage{luatex85}

\documentclass[11pt,a4paper]{article}
\usepackage{iftex}

\ifPDFTeX
\usepackage[utf8]{inputenc}
\usepackage[T1]{fontenc}
\fi

\usepackage{amsmath}
\usepackage{amssymb}
\usepackage{amsthm}
\usepackage{mathtools}
\usepackage{thmtools}
\usepackage{isomath}
\usepackage{dsfont}

\ifPDFTeX
\usepackage{libertine}
\usepackage[libertine]{newtxmath} 
\usepackage[scaled=0.96]{zi4} 
\else
\usepackage{fontspec}
\usepackage{unicode-math}
\fi

\usepackage[DIV=10]{typearea} 

\usepackage{microtype}

\usepackage[ruled,vlined,linesnumbered,noend]{algorithm2e}

\usepackage{xcolor}

\usepackage{lineno}

\usepackage[
	style=alphabetic,
    maxalphanames=4,
    minalphanames=4,
    maxcitenames=4,
    minnames=1,
    maxbibnames=20,
	backref=true,
	doi=true,
	url=true,
	backend=bibtex8,
]{biblatex}

\usepackage[ocgcolorlinks]{hyperref} 

\input{commands}

\usepackage[textwidth=2cm,textsize=footnotesize]{todonotes}

\colorlet{DarkRed}{red!50!black}
\colorlet{DarkGreen}{green!50!black}
\colorlet{DarkBlue}{blue!50!black}

\hypersetup{
	linkcolor = DarkRed,
	citecolor = DarkGreen,
	urlcolor = DarkBlue,
	bookmarksnumbered = true,
	linktocpage = true
}

\DontPrintSemicolon
\SetKwProg{Procedure}{Procedure}{}{}
\SetProcNameSty{textsc}
\SetFuncSty{textsc}
\SetKw{KwAnd}{and}
\SetKw{KwDo}{do}

\SetCommentSty{mycommfont}

\SetKwFunction{Spanner}{Spanner}
\SetKwFunction{Connect}{Connect}
\SetKwFunction{BundleSpanner}{BundleSpanner}
\SetKwFunction{SpectralSparsify}{SpectralSparsify}
\SetKwFunction{SpectralSparsifyapriori}{SpectralSparsify-apriori}
\SetKwFunction{LPSolve}{LPSolve}
\SetKwFunction{ComputeLeverageScores}{ComputeLeverageScores}
\SetKwFunction{ComputeApxWeights}{ComputeApxWeights}
\SetKwFunction{ComputeInitialWeights}{ComputeInitialWeights}
\SetKwFunction{median}{median}
\SetKwFunction{PathFollowing}{PathFollowing}
\SetKwFunction{CenteringInexact}{CenteringInexact}
\SetKwFunction{ProjectMixedBall}{ProjectMixedBall}

\declaretheorem[numberwithin=section]{theorem}
\declaretheorem[numberlike=theorem]{lemma}

\declaretheorem[numberlike=theorem]{corollary}
\declaretheorem[numberlike=theorem]{definition}

\newcommand*{\vect}[1]{\mathbf{#1}}

\newcommand{\x}{\vect{x}}

\DeclareMathOperator{\poly}{poly}

\usetikzlibrary{shapes.geometric, arrows, positioning, shapes.multipart}

\DeclareMathOperator{\ID}{ID}
\let\P\relax
\DeclareMathOperator{\P}{\mathbb{P}}
\DeclareMathOperator{\Z}{\mathbb{Z}}
\DeclareMathOperator{\N}{\mathbb{N}}
\DeclareMathOperator{\R}{\mathbb{R}}
\DeclareMathOperator{\sign}{sign}
\DeclareMathOperator{\rank}{rank}
\DeclareMathOperator{\diag}{diag}
\DeclareMathOperator{\dom}{dom}

\DeclareMathOperator*{\argmax}{arg\ max}
\DeclareMathOperator*{\argmin}{arg\ min}

\DeclareMathOperator{\Ot}{\widetilde{\mathit{O}}}

\SetAlgoSkip{bigskip}

\title{The Laplacian Paradigm in the Broadcast Congested Clique}

\author{
  Sebastian Forster\thanks{Department of Computer Science, University of Salzburg, Austria}
  \and
  Tijn de Vos\samethanks
}

\date{}

\hypersetup{
	pdftitle = {The Laplacian Paradigm in the Broadcast Congested Clique},
	pdfauthor = {Tijn de Vos, Sebastian Forster}
}

\begin{document}
\maketitle
\thispagestyle{empty}
\begin{abstract}
In this paper, we bring the main tools of the Laplacian paradigm to the Broadcast Congested Clique. We introduce an algorithm to compute spectral sparsifiers in a polylogarithmic number of rounds, which directly leads to an efficient Laplacian solver. 
Based on this primitive, we consider the linear program solver of Lee and Sidford~\cite{LS14}. We show how to solve certain linear programs up to additive error~$\epsilon$ with $n$ constraints on an $n$-vertex Broadcast Congested Clique network in $\Ot(\sqrt{n}\log(1/\epsilon))$ rounds. Using this, we show how to find an exact solution to the minimum cost flow problem in $\Ot(\sqrt{n})$ rounds.
\end{abstract}

\newpage
\tableofcontents
\newpage

\input{introduction}

\section{Preliminaries}
First we detail the models we will be working with. Next, we review spanners and sparsifiers, and how to construct the latter from the former. Then we show how spectral sparsifiers can be used for solving Laplacian systems. Finally, we introduce flow problems on weighted graphs. 

\subsection{Models}
In this paper, we consider multiple variants of message passing models with bandwidth constraints on the communication. Let us start by defining the CONGEST model. The CONGEST model~\cite{Peleg00} consists of a network of processors, which communicate in synchronous rounds. In each round, a processor can send information to its neighbors over a non-faulty link with limited bandwidth. We model the network of processors by a graph $G=(V,E)$, where we identify the processors with the vertices and the communication links with the edges. We write $n=|V|$ and $m=|E|$. Each vertex has a unique identifier of size $O(\log n)$, initially only known by the vertex itself and its neighbors. Computation in this model is done in rounds. At the start of each round, each vertex can send one message to each of its neighbors, and receives messages from them. The messages are of size at most $B=\Theta(\log n)$. Before the next round, each vertex can perform (unlimited) internal computation. We measure the efficiency of an algorithm by the number of rounds. 

In the CONGEST model, each vertex can send distinct messages to each of its neighbors. A more strict assumption on message passing, is that each vertex sends the same message to each of its neighbors, essentially broadcasting it to its neighbors. The CONGEST model together with this assumption is called the \emph{Broadcast CONGEST model}~\cite{Lynch96}. 

Alternatively, we can let the communication network be independent of the graph being studied. More precisely, we allow communication between each pair of vertices. Together with the bandwidth constraint, this is called the \emph{Congested Clique}~\cite{LPSPP05}. If we also impose the broadcast constraint, we have the \emph{Broadcast Congested Clique}~\cite{DKO12}.  

\subsection{Spanners and Spectral Sparsification}\label{sc:prelim_spanners_sparsifiers}
The \emph{Laplacian matrix} of a weighted graph $G$, or the \emph{graph Laplacian}, is a matrix $L\in \R^{n\times n}$ defined by 
\begin{align*}
    L_{uv} = \begin{cases}  -w(u,v) & \text{if $u$ is adjacent to $v$};\\
    \sum_{x\in V} w(u,x) &\text{if }u=v;\\
    0 &\text{else.}
    \end{cases}
\end{align*}
Alternatively, we can define the Laplacian matrix in terms of the \emph{edge-vertex incidence matrix} $B$, defined by 
\begin{align*}
    B(e,v) := \begin{cases}
    1 & \text{if }v=e_{\rm{head}};\\
    -1 &\text{if } v=e_{\rm{tail}};\\
    0 &\text{otherwise.}
    \end{cases}
\end{align*}
The Laplacian then becomes $L=B^TWB$, where $W\in \R^{m\times m}$ is the diagonal matrix defined by the weights: $W_{ee}:=w(e)$. 

Spectral sparsifiers were first introduced by Spielman and Teng \cite{ST11}. A spectral sparsifier is a (reweighted) subgraph that has approximately the same Laplacian matrix as the original graph.
\begin{definition}
\label{def:specspars}
Let $G=(V,E)$ be a graph with weights $w_G\colon E\to \R$, and $n=|V|$. We say that a subgraph $H\subseteq G$ with weights $w_H\colon E(H)\to \R$ is a $(1\pm\epsilon)$-\emph{spectral sparsifier} for $G$ if we have for all $x\in\R^n$:
\begin{align} (1-\epsilon) x^TL_Hx \leq x^T L_G x\leq (1+\epsilon)x^T L_Hx, \label{eq:specspars}\end{align}
where $L_G$ and $L_H$ are the Laplacians of $G$ and $H$ respectively. 
\end{definition}

We introduce the short-hand notation $A \preccurlyeq B$ when $B-A$ is positive semi-definite. This reduces equation~\ref{eq:specspars} to $(1-\epsilon)L_H \preccurlyeq L_G \preccurlyeq (1+\epsilon)L_H$. 

Koutis and Xu~\cite{KX16} showed how to compute a spectral sparsifier by repeatedly computing spanners. This technique was later slightly improved by Kyng et al.~\cite{KPPS17}.
Spanners are a special type of spanning subgraphs, where we demand that distances are preserved up to a constant factor. Trivially, any graph is a spanner of itself. In practice, the goal will be to find sparse subgraphs that are still spanners for the input graph. 

\begin{definition}
Let $G=(V,E)$ be a graph with weights $w\colon E\to \R$. We say that a subgraph $S\subseteq G$ with weights $w_S=w|_S$ is a \emph{spanner} of stretch $\alpha$ for $G$ if for each $u,v \in V$ we have 
\[ d_S(u,v)\leq \alpha d_G(u,v),\]
where we write $d_H(u,v)$ for the distance from $u$ to $v$ in $H$. A \emph{$t$-bundle spanner} of stretch $\alpha$ is a union $T=\bigcup_{i=1}^k T_i$, where each $T_i$ is a spanner of stretch $\alpha$ in $G\setminus \bigcup_{j=1}^{i-1}T_j$.  
\end{definition}

The algorithm of Koutis and Xu is relatively simple: compute a $t$-bundle spanner of stretch $\alpha$, sample the remaining edges with probability $1/4$, repeat for $\lceil\log(m)\rceil$ iterations on the computed bundle spanner and sampled edges. The sparsifier then consists of the last bundle spanner, together with the set of edges left after the $\lceil \log(m)\rceil$ iterations, where edges are reweighted in a certain manner. In the original algorithm, the stretch $\alpha$ was fixed, but the number $t$ of spanners in each bundle grew in each iteration. Kyng et al.~\cite{KPPS17} showed that $t$ can be kept constant throughout the algorithm, leading to a sparser result. 

\begin{algorithm}[H]
\SetAlgoLined \caption{\textsc{SpectralSparsifyOutline}($V,E,w,\epsilon$)}
\label{alg:SS_outline}
Set $k := \lceil\log n\rceil$, $t:= 400\log^2(n) \epsilon^{-2}$, and $E_0:= E$.\\
\For{$i=1,\dots,  \lceil \log m\rceil$}{
Compute a $t$-bundle spanner $B_i$ of stretch $k$.\\
$E_i := B_i$.\\
\ForEach{$e\in E_{i-1}\setminus B_i$}{
With probability $1/4$: $E_i \leftarrow E_i \cup \{e\}$ and $w(e)\leftarrow 4w(e)$.}
}
\Return{$\left(V,E_{\lceil \log m\rceil}\right)$}.
\end{algorithm}

\subsection{Laplacian Solving}\label{sc:laplacian} 
We consider the following problem. Let $L_G$ be the Laplacian matrix for some graph $G$ on $n$ vertices. Given $b\in \mathbb R^n$, we want to solve $L_Gx=b$. Solving Laplacian equation exactly can be computationally demanding. 
Therefore, we consider an approximation to this problem: we want to find $y\in R^n$ such that $||x-y||_{L_G}\leq \epsilon||x||_{L_G}$, where we write $||x||_M:=\sqrt{x^TMx}$ for any $M\in\R^{n\times n}$. One way to approach this is by using a spectral sparsifier of $G$. Hereto we use \emph{preconditioned Chebyshev iteration}, a well known technique from numerical analysis~\cite{Axelsson96,Saad03}. The statement below most closely resembles the formulation of Peng~\cite{peng13}.  
\begin{theorem}\label{thm:PCI}
Suppose we have symmetric positive semi-definite matrices $A, B\in \R^{n\times n}$, and a parameter $\kappa\geq 1$ satisfying
\[A \preccurlyeq B \preccurlyeq \kappa A.\]
Then there exists an algorithm that, given a vector $b\in \R^n$ and parameter $\epsilon\in (0,\frac{1}{2}]$, returns a vector $y\in \R^n$ such that 
\[ ||x-y||_A \leq \epsilon ||x||_A,\]
for some $x\in \R^n$ satisfying $Ax=b$. The algorithm takes $O(\sqrt{\kappa}\log (1/\epsilon))$ iterations, each
consisting of multiplying $A$ by a vector, solving a linear system involving $B$, and a constant number of vector operations.
\end{theorem}

This yields the following corollary for Laplacian solving using spectral sparsifiers. 

\begin{corollary}
\label{cor:laplaciansolving}
Let $G$ be a weighted graph on $n$ vertices, let $\epsilon\in(0,\frac{1}{2}]$ be a parameter, and let $b\in \R^n$ a vector. Suppose $H$ is a $\left(1\pm\frac{1}{2}\right)$-spectral sparsifier for $G$. Then there exists an algorithm that outputs a vector $y\in \R^n$ such that $||x-y||_{L_G}\leq \epsilon||x||_{L_G}$, for some $x\in R^n$ satisfying $L_Gx=b$. The algorithm takes $O(\log(1/\epsilon))$ iterations, each consisting of a multiplying $L_G$ by a vector, solving a Laplacian system involving $L_H$, and a constant number of vector operations. 
\end{corollary}
\begin{proof}
As $H$ is a sparsifier for $G$, we have: $\left(1-\frac{1}{2}\right)L_H \preccurlyeq L_G \preccurlyeq \left(1+\frac{1}{2}\right)L_H$, which we can rewrite to 
\[ L_G \preccurlyeq \left(1+\frac{1}{2}\right)L_H \preccurlyeq \frac{1+\frac{1}{2}}{1-\frac{1}{2}}L_G.\]
We set $A:=L_G$ and $B:=\left(1+\frac{1}{2}\right)L_H$, which are clearly both symmetric positive semi-definite. Furthermore, we set $\kappa:= \frac{1+\frac{1}{2}}{1-\frac{1}{2}}=3$. We apply Theorem~\ref{thm:PCI} with these settings to obtain the result. 
\end{proof}

\subsection{Flow Problems}\label{sc:prelim_flow}
In this section we formally define the maximum flow and the minimum cost maximum flow problems. Let $G=(V,E)$ be a directed graph, with capacities $c\colon E\to \Z_{\geq0}$, and designated source and target vertices $s,t\in V$. We say $f\colon E \to \R_{\geq 0}$ is an $s$-$t$ \emph{flow} if
\begin{enumerate}
	\item for each vertex $v\in V\setminus\{s,t\}$ we have $\sum_{e\in E : v\in e} f_e =0$;
	\item for each edge $e\in E$ we have $f_e \leq c_e$. 
\end{enumerate}
The value of the flow $f$ is defined as $\sum_{u: (s,u)\in E}f_{(s,u)}$. The maximum flow problem is to find a flow of maximum value. 
Additionally, we can have costs on the edges: $q\colon E \to \Z_{\geq 0}$. The cost of the flow $f$ is defined as $\sum_{e\in E} q_ef_e$. The minimum cost maximum flow problem is to find a flow of minimum cost among all flows of maximum value. 

Both problems allow for a natural linear program formulation. We present one for the minimum cost maximum flow problem, as this is the more general problem. 
Denote $B$ for the edge-vertex incidence matrix (see Section~\ref{sc:prelim_spanners_sparsifiers}). Then we can write this as:
\begin{align*}
	\min_{0\leq x\leq c} q^T x \text{ such that }Bx=Fe_t-Fe_s,
\end{align*}
for $F$ the value of the maximum flow, and $e_s$ and $e_t$ the vectors defined by $(e_i)_j:=\delta_{ij}$. The answer to the minimum cost maximum flow problem is then found by a binary search over $F$.

\section{Spectral Sparsifiers and Laplacian Solving}

In this section, we show how to construct spectral sparsifiers in the Broadcast CONGEST model, so in particular also for the Broadcast Congested Clique. We do this following the method of Koutis and Xu~\cite{KX16}, which consists of repeatedly computing spanners and sampling the remaining vertices, see Section~\ref{sc:prelim_spanners_sparsifiers}. While sampling edges is easy in the CONGEST model, it is highly non-trivial in the Broadcast CONGEST model. The reason for this is that in the CONGEST model the sampling of an edge can be done by one endpoint, and communicated to the other endpoint. In the Broadcast CONGEST model, the sampling can be done by one endpoint, but the result cannot be communicated efficiently to the other endpoint due to the broadcast constraint. To circumvent this, we show that the sampling needed for spectral sparsification can be done on the fly, rather than a priori in each iteration. Moreover, we show the result can be communicated implicitly. In Section~\ref{sc:prob_spanners}, we show how to compute spanners where we have probabilities on edges existing, whether an edge exists is evaluated on the fly and (implicitly) communicated to the other endpoint. In Section~\ref{sc:BC_sparsifier} we show how to use this spanner construction to compute spectral sparsifiers in the Broadcast CONGEST model. 

\subsection{Spanners with Probabilistic Edges}\label{sc:prob_spanners}

Our goal is to compute a $(2k-1)$-spanner for a given probabilistic graph. More precisely, let $G=(V,E,w)$ be an undirected, weighted graph on $n$ vertices, with $p\colon E \to [0,1], e\mapsto p_e$ a probability function on the edges, and $k\leq n$ the parameter for the stretch of the spanner. We will give an algorithm \Spanner{$V$,$E$,$w$,$p$,$k$} that computes a subset $F\subseteq E$, and divides this into two sets $F=F^+\sqcup F^-$, such that each edge $e\in F$ is part of $F^+$ independently with probability $p_e$. This results in a $(2k-1)$-spanner $S= (V,F^+)$ for all graphs $(V,F^+\cup E'')$, where $E''\subseteq E\setminus F$. 
Since this is a distributed algorithm, the output comes in a local form. At the end, each vertex $v$ has identified $F^+_v$ and $F^-_v$, where $u\in F^\pm_v \iff (u,v)\in F^\pm$. \\

When $p\equiv1$, our algorithm essentially reduces to the algorithm of Baswana-Sen from \cite{BS07}. All computational steps coincide, but a difference in communication remains. The reason hereto is that in our algorithm the weights of edges are included in the communication. Depending on the magnitude of the weights, this can result in multiple rounds for each message, and consequently more rounds in total. 

For the presentation of Baswana and Sen's algorithm, we follow the equivalent formulation of Becker et al.~\cite{BeckerFKL21}, which can be found in Appendix \ref{app:BSalg}. The general idea is that clusters are formed and revised through a number of phases. In each phase, a few of the existing clusters are sampled. These clusters move onto the next phase. Vertices from an unsampled cluster try to connect to a sampled cluster and to some neighboring clusters. As edges only exist with a certain probability, they need to be sampled before they can be used. We will make sure that the two vertices adjacent to an edge, never try to use it at the same time. 
When a vertex has tried to use an edge, the edge will always be broadcasted if it exists. If not, it turns out that the other vertex adjacent to this edge will be able to deduce this, without it being communicated explicitly.  

Whenever we speak of the neighbors of a vertex $v$, denoted by $N_v$, we mean all neighbors that do not lie in the set $F^-_v$ of `deleted neighbors': $N_v:= \{ u\ |\ (u,v)\in E\setminus F^-_v\}$. Note that this set of neighbors will be subject to change throughout the process, as the number of elements in $F^-_v$ grows. \\

\textbf{Step 1: Cluster marking}\\
Initially, each vertex is a singleton \emph{cluster}: $R_1:=\{\{v\}\ |\ v\in V\}$. The main part of the algorithm will be $k-1$ `phases', indexed $i=1,\dots k-1$. In phase $i$, the center of each cluster (the first vertex in the cluster) \emph{marks} the cluster with probability $n^{-1/k}$ and broadcasts this result to the cluster. These clusters will move on to the next phase: $R_{i+1}$ is defined to be the set of clusters marked in phase $i$. We define the identifier $\ID(X)$ of a cluster $X$ to equal the ID of the center of the cluster. 
Each phase consists of cluster marking, followed by steps 2 and 3.\\

\textbf{Step 2: Connecting to marked clusters}\\
Let $v$ be a vertex in an unmarked cluster $X_v\in R_i$. The first thing $v$ does, is trying to connect to one of the marked clusters. It does this using the procedure \Connect. Hereto we define $N$ to be the set of all neighbors of $v$ which lie in a marked cluster: $N := N_v \cap \bigcup_{X\in R_{i+1}}X$. Now we let $(u,N^-_v):=\Connect{$N$,$p|_{N}$}$. Note that if $N=\emptyset$, \Connect returns $(u,N^-_v)=(\bot,\emptyset)$. If $u=\bot$, we broadcast $\left(\bot,W^{(i)}_v:=\infty\right)$. If it returns $u\neq \bot$, we add $u$ to $F^+_v$, $v$ joins the cluster $X_u$ of $u$ (it stores this decision by saving $\ID(X_u)$), and we broadcast $\left(\ID(X_u),u,W^{(i)}_v:=w(u,v)\right)$. In both cases, we add $N^-_v$ to $F^-_v$.

\begin{algorithm}[H]
\SetAlgoLined \caption{\textsc{Connect}($N,p$)}
Sort $N$ ascendingly according to the weight of the corresponding edge. In case of equal weights, the neighbor with the smallest ID comes first. Write $N=\{x_1, \dots, x_{|N|}\}$.\\
$u := \bot$.\\
$N^-:= \emptyset$.\\
$i :=1$.\\
\While{$i \leq |N|$ and $u=\bot$}{
Sample $r\in[0,1]$ uniformly at random. \\
\eIf{$r\leq p_{(x_i,v)}$}
{$u\leftarrow x_i$.}
{$N^-\leftarrow N^-\cup\{x_i\}$.}
$i\leftarrow i+1.$
}
\Return{$(u,N^-)$}.
\end{algorithm}

After this step, all vertices $v$ in unmarked clusters may have joined marked clusters, and they have updated their sets $F^+_v$ by adding $u$, and $F^-_v$ by adding $N^-_v$. We also want to propagate these updates in $F^\pm_v$ to the neighbors of $v$. This is easy for $F^+_v$, since we can broadcast $u$. However, we do not want to broadcast the set $N^-_v$, as it can be large. Instead we make use of the choices in \Connect to communicate changes in $F^-_v$ implicitly. 

Let $u$ be a neighbor of $v$ in a marked cluster. If $v$ has broadcasted $\left(\ID(X_u),u,W^{(i)}_v:=w(u,v)\right)$, then $u$ adds $v$ to $F^+_u$. There are three situations where $u$ adds $v$ to $F^-_u$:
\begin{enumerate}
    \item If $v$ broadcasted $\left(\bot,W^{(i)}_v\right)$;
    \item If $v$ broadcasted $\left(\ID(X_{u'}),u',W^{(i)}_v=w(u',v)\right)$ with $w(u',v)>w(u,v)$;
    \item If $v$ broadcasted $\left(\ID(X_{u'}),u',W^{(i)}_v:=w(u',v)\right)$ with $w(u',v)=w(u,v)$ and $\ID(u')>\ID(u)$.
\end{enumerate}
In any other case, $u$ does nothing. This step ensures that $v$ gets added to $F^-_u$ if and only if $u\in N^-_v$. In total, this results in $u\in F^\pm_v \iff (u,v)\in F^\pm$ for all vertices $u,v\in V$.\\

As a final note: each vertex has broadcasted the ID of the cluster it joins, its neighbors keep track of these changes, as they will need the new cluster IDs when they try to connect to a marked cluster in the \emph{next} phase. For the remainder of \emph{this} phase (step 3), the `old' cluster IDs are still valid.\\

\textbf{Step 3: Connections between unmarked clusters}\\
In this step, we create connections between the unmarked clusters. In the previous part, the situation was asymmetric: vertices of unmarked clusters connected to vertices in marked clusters.  To make sure that at most one vertex decides upon the existence of an edge, we create two substeps. In the first substep a vertex $v$ in cluster $X_v\in R_i$ can only connect to a neighboring cluster $X$ if $\ID(X)<\ID(X_v)$. In the second substep, a vertex $v$ can only connect to neighboring clusters with higher ID. This way all necessary connections can be made, while no two vertices will simultaneously try to decide on the existence of the edge between them. \\

\textbf{Step 3.1: Connecting to a cluster with a smaller ID}\\
Let $v$ be a vertex in an unmarked cluster $X_v\in R_i\setminus R_{i+1}$. We will try to connect to each neighboring cluster $X\in R_i \setminus R_{i+1}$ with $\ID(X)<\ID(X_v)$. Fix such a cluster $X$. Let $N$ be the neighbors $u$ of $v$ in this cluster, with $w(u,v)<W^{(i)}_v$, i.e.\ $N:= \{u\in N_v\cap X: w(u,v)<W^{(i)}_v\}$. Similar as before, we run \Connect to decide which neighbor to connect to: $(u,N^-_v):=\Connect(N,p|_N)$. If it returns $u\neq \bot$, we add $u$ to $F^+_v$ and we broadcast $\left(\ID(X),u,w(u,v)\right)$. If \Connect returns $\bot$, we simply broadcast $\left(\ID(X),\bot\right)$. In both cases we add $N^-_v$ to $F^-_v$. Again we wish to propagate these updates to $v$'s neighbors. As before, we communicate this implicitly. 

Let $u$ be a vertex in neighboring cluster $X$ with $\ID(X)<\ID(X_v)$ and $w(u,v)<W^{(i)}_v$. If $v$ has broadcasted $\left(\ID(X),u,w(u,v)\right)$, then $u$ adds $v$ to $F^+_{u}$. Again, there are three situations where $u$ adds $v$ to $F^-_{u}$:
\begin{enumerate}
    \item If $v$ broadcasted $(\ID(X),\bot)$;
    \item If $v$ broadcasted $\left(\ID(X),u',w(u',v)\right)$ with $w(u',v)>w(u,v)$;
    \item If $v$ broadcasted $\left(\ID(X),u',w(u',v)\right)$ with $w(u',v)=w(u,v)$ and $\ID(u')>\ID(u)$.
\end{enumerate}
In any other case, $u$ does nothing. As before, note that this step ensures that $u\in F^\pm_v \iff (u,v)\in F^\pm$ for all vertices $u,v\in V$.\\

\textbf{Step 3.2: Connecting to a cluster with a bigger ID}\\
Vertices $v$ in an unmarked cluster $X_v$ have now connected to neighboring unmarked clusters $X$ with $\ID(X)<\ID(X_v)$ and the sets $F^\pm_v$ have been updated accordingly. However, we need to connect to \emph{all} unmarked neighboring clusters, just as in the original algorithm (as depicted in Appendix \ref{app:BSalg}). Therefore we move on to the neighboring clusters $X$ with $\ID(X)>\ID(X_v)$. The process for these clusters is completely analogous to substep 3.1, and thus will not be given here. \\

\textbf{Step 4: After the $k-1$ phases}\\
In the last part of the algorithm, we want to connect each vertex $v$ to all its neighboring clusters in $R_k$. This is again done in three steps, similar to the steps 2, 3.1, and 3.2 in the phases above.
\begin{itemize}
    \item[4.1] All vertices $v\notin \bigcup_{X\in R_k}X$ that are not part of any remaining cluster connect, using \Connect$\left(N_v\cap X',p|_{N_v\cap X'}\right)$, to each neighboring remaining clusters $X'\in \bigcup_{X\in R_k}X$. As before, they broadcast how they connect such that vertices $u\in \bigcup_{X\in R_k}X$ in remaining clusters can add edges to $F^\pm_u$ accordingly.
    \item[4.2] Vertices $v\in X_v \in R_k$ connect, using \Connect$\left(N_v\cap X,p|_{N_v\cap X}\right)$, to each neighboring remaining clusters $X\in R_k$ with $\ID(X)<\ID(X_v)$. As before, they broadcast the result, such that neighbors $u$ can add edges to $F^\pm_u$ accordingly.
    \item[4.3] Vertices $v\in X_v \in R_k$ connect, using \Connect$\left(N_v\cap X,p|_{N_v\cap X}\right)$, to each neighboring remaining clusters $X\in R_k$ with $\ID(X)>\ID(X_v)$. As before, they broadcast the result, such that neighbors $u$ can add edges to $F^\pm_u$ accordingly.
\end{itemize}

In the following lemma we show that this algorithm indeed gives a spanner of stretch $2k-1$. 

\begin{lemma}\label{lm:spanner}
The spanner $S=(V,F^+)$ has stretch at most $2k-1$ for all graphs $(V,F^+\cup E'')$, where $E''\subseteq E\setminus F$. For any choice of $E$, it has at most $|F^+|=O\left(kn^{1+1/k}\right)$ edges in expectation. Moreover, we obtain an orientation on $F^+$ such that each edge has out-degree $O(kn^{1/k})$ in expectation. 
\end{lemma}
\begin{proof}
First of all, note that setting $p\equiv 1$ reduces this more involved algorithm to the original algorithm, given in Appendix \ref{app:BSalg}, which we know to correctly create a spanner. We claim \Spanner{$V,F^+\cup E'',w|_{F^+\cup E''},1,k$} also outputs $(V,F^+)$ as spanner, under the following assumption on the marking of clusters. In step 1, each cluster marks itself with probability $n^{-1/k}$. We can imagine that it does this by drawing from some source of random bits. Our assumption is that these random bits are the same for both algorithms. This assumption can be made, since these bits are independent of the probability $p$ on the edges. From now on, we call \Spanner{$V,E,w,p,k$} algorithm $A$ and \Spanner{$V,F^+\cup E'',w|_{F^+\cup E''},1,k$} algorithm $B$. We claim that if algorithm $A$ outputs $(F^+,F^-)$, and $E''\subseteq E\setminus (F^+\cup F^-)$, that using $F^+\cup E''$ as its input, algorithm $B$ will output $(F^+,\emptyset)$. Since we already know that the output of algorithm $B$ gives a spanner for $(V,F^+\cup E'')$, this proves the lemma.

We will not only show that the output of the two algorithms is the same. We will even show that all intermediate steps (creating clusters and selecting spanner edges) are the same. We will prove this claim by induction. It is clear that the initialization of both algorithms is the same. We need to show that if both algorithms have produced the same situation up to a certain point, the next decision will also be the same. These decisions take place whenever a vertex tries to connect to some cluster. This happens in steps 2, 3.1, 3.2, 4.1, 4.2, and 4.3. Every time, the same principle is applied. We will give the proof of the induction procedure at step 2. 

We assume so far the created clusters are exactly the same. Suppose $v$ is part of some unmarked cluster $X_v$. We investigate what the \Connect procedure results in for the two different algorithms. Suppose \Connect outputs $(\bot,N^-_v=N)$ in algorithm $A$. That means all neighbors $u\in N$ of $v$ end up in $F^-_v$. Hence $v$ has no neighbors in $(V,F^+\cup E'')$, as $F^-\cap (F^+\cup E'')=\emptyset$. Therefore algorithm $B$ will output $(\bot,\emptyset)$. 

Now suppose \Connect outputs $(u,N^-_v)$ in algorithm $A$. For contradiction, suppose that algorithm $B$ outputs $u'\neq u$. When algorithm $B$ calls the procedure \Connect$(N,1)$, we know $u\in N$, as it is a neighbor. We note that \Connect sorts $N$ ascendingly according to weights, and in case of equal weights the smallest ID comes first. Since $p\equiv1$, the first option is accepted. So $u'$, must come before $u$. Meaning that $w(u',v)<w(u,v)$, or $w(u',v)= w(u,v)$ and $\ID(u')<\ID(u)$. In both cases, $u'$ also comes before $u$ when algorithm $A$ runs \Connect. Since algorithm $A$ did not accept $u'$, this implies that $u'\in F^-_v$. That means $(u',v)\in F^-$, thus $u'$ is not a neighbor of $v$ in $(V,F^+\cup E'')$; a contradiction. 

Similar arguments hold for all other indicated steps. We conclude that both algorithms output the same graph. Baswana and Sen~\cite[Theorem 4.3]{BS07} show that this is a $(2k-1)$-spanner for $(V,F^+\cup E'')$ and that it has at most $O\left(kn^{1+1/k}\right)$ edges in expectation.

For the orientation, we simply orient edges within a cluster from child to parent. We orient edges between clusters from the vertex that added it to the other vertex. If both endpoint of an edge want to add the edge, we orient it arbitrarily. According to Baswana and Sen, each vertex adds $O(kn^{1/k})$ edges in expectation, giving the result. 
\end{proof}

Next, we analyze the running time of the algorithm. 

\begin{lemma}\label{lm:spanner_runningtime}
The algorithm \Spanner{$V,E,w,p,k$} takes $O\left(kn^{1/k}\left(\log n+\log W\right)\right)$ rounds.
\end{lemma}
\begin{proof}
The algorithm consists of $k-1$ phases, consisting of step 1, 2, and 3, and a final step 4. In step 1, the center needs to broadcast the result of the marking to all vertices in its cluster. This takes at most $k-1$ rounds, as the cluster is a tree of depth at most $k-1$. In step 2 there is only one message: vertices in unmarked clusters announce which marked cluster they join (if any), by broadcasting the ID of the vertex they are connecting to and the weight of the corresponding edge. This takes $1+\frac{\log W}{\log n}$ rounds. In step 3, each vertex broadcasts the edges added to the spanner and the corresponding weights, taking $1+\frac{\log W}{\log n}$ rounds per edge. Clearly the number of edges added in each phase is bounded by the total number of added edges. The latter is $O(n^{1/k})$ in expectation and $O\left(n^{1/k}\log n\right)$ with high probability. Step 4 is adheres the same upper bound as step 3. 

Adding all of this together, we obtain $k-1$ phases, each consisting of at most $(k-1)+1+O\left(n^{1/k}\log n\right)\left(1+\frac{\log W}{\log n}\right)$ rounds, and a final step of at most $O\left(n^{1/k}\log n\right)\left(1+\frac{\log W}{\log n}\right)$ rounds. This results in a total of at most $O\left(kn^{1/k}\left(\log n+\log W\right)\right)$ rounds.
\end{proof}

We end this section with the following straight forward algorithm to compute a bundle of spanners. \\
\begin{algorithm}[H]
\SetAlgoLined \caption{\textsc{BundleSpanner}($V,E,w,p,k,t$)}
Set $E_0:= E$.\\
\For{$i=1,\dots, t$}{
    $(F^+_i,F^-_i) \leftarrow$\Spanner{$V,E_{i-1},w|_{E_{i-1}},p|_{E_{i-1}},k$}.\\
    $E_i \leftarrow E_{i-1}\setminus\left(F^+_i\cup F^-_i\right)$.
}
$B \leftarrow \bigcup_{i=1}^t F^+_i$.\\
$C \leftarrow \bigcup_{i=1}^t F^-_i$.\\
\Return{$(B,C)$}.
\end{algorithm}
By Lemma~\ref{lm:spanner}, this algorithm produces a $t$-bundle $B$ of $(2k-1)$-spanners, where $|B|=O(tkn^{1+1/k})$. By Lemma~\ref{lm:spanner_runningtime}, it takes a total of $O\left(tkn^{1/k}\left(\log n+\log W\right)\right)$ rounds.

\subsection{Sparsification}\label{sc:BC_sparsifier}
The algorithm we give for spectral sparsification is based upon Algorithm~\ref{alg:SS_outline}, as given in Section~\ref{sc:prelim_spanners_sparsifiers}. Below, in Algorithm \ref{alg:SSapriori}, we give a more concrete version of this algorithm, specifying how to compute the bundle spanner. This algorithm repeatedly calculates a $t$-bundle spanner, and adds the remaining edges with probability $1/4$. We amend this algorithm to be able to apply it in the Broadcast CONGEST model. The key difference is that whenever we need to keep edges with probability $1/4$ we do this ad hoc and `locally', rather than a priori and `central'. 

Kyng et al.~\cite{KPPS17} have shown that the number $t$ of spanners in each bundle can be kept the same throughout the algorithm, as opposed to increasing it in each iteration, which is done in the original algorithm of Koutis and Xu~\cite{KX16}. This results into a reduction of $\log n$ in the size of the spanner.  \\
\begin{algorithm}[H]
\SetAlgoLined \caption{\textsc{SpectralSparsify-apriori}($V,E,w,\epsilon$)}
\label{alg:SSapriori}
Set $k := \lceil\log n\rceil$, $t:= 400\log^2(n) \epsilon^{-2}$, and $E_0:= E$.\\
\For{$i=1,\dots,  \lceil \log m\rceil$}{
$(B_i,C_i) \leftarrow$\BundleSpanner{$V,E_{i-1},w|_{E_{i-1}},p\equiv1,k,t$}.\\
$E_i := B_i$.\\
\ForEach{$e\in E_{i-1}\setminus B_i$}{
With probability $1/4$: $E_i \leftarrow E_i \cup \{e\}$ and $w(e)\leftarrow 4w(e)$.}
}
\Return{$\left(V,E_{\lceil \log m\rceil}\right)$}.
\end{algorithm}

We use the spanner construction given in the previous section, which incorporates the ad hoc sampling with the spanner construction. \\
\begin{algorithm}[H]
\SetAlgoLined \caption{\textsc{SpectralSparsify}($V,E,w,\epsilon$)}
Set $k := \lceil\log n\rceil$, $t:= 400\log^2(n) \epsilon^{-2}$, and $E_0:= E$.\\
Define $p\colon E\to [0,1]$, by $p\equiv 1$.\\
\For{$i=1,\dots, \lceil \log m\rceil$}{
$(B_i,C_i) \leftarrow$\BundleSpanner{$V,E_{i-1},w|_{E_{i-1}},p|_{E_{i-1}},k,t$}.\\
$E_i \leftarrow E_{i-1}\setminus C_i$.\\
\ForEach{$e\in B_i$}{
$p(e)\leftarrow 1$.}
\ForEach{$e\in E_i\setminus B_i$}{
$p(e)\leftarrow p(e)/4$.\\
$w(e) \leftarrow 4w(e)$.}
}
Set $E':=B_{ \lceil \log m\rceil}$.\\
At each vertex $v\in V$: \\
\For{$(u,v)\in E_{ \lceil \log m\rceil}\setminus E'$}{
\If{$\ID(v)<\ID(u)$}{With probability $p(u,v)$: add $(u,v)$ to $E'$ and broadcast $(u,v)$.}
}
\Return{$(V,E')$}.
\end{algorithm}

For correctness, we relate the output of our sparsification algorithm, to the output of the sparsification algorithm from Koutis and Xu~\cite{KX16}, where we use the improved version of Kyng et al.~\cite{KPPS17} with fixed $t$. \\

\begin{lemma}
\label{lm:SS_alg_eq}
Given any input graph $G=(V,E,w)$, and any possible graph $H=(V,E_H,w)$, we have that 
\begin{align*}
    \P\left[\SpectralSparsify{$V,E,w,\epsilon$}=H\right]=\P\left[\SpectralSparsifyapriori{$V,E,w,\epsilon$}=H\right].
\end{align*}
\end{lemma}
\begin{proof}
Throughout this proof, we will use superscripts $(\mathrm{ap})$ for the setting with a priori sampling and $(\mathrm{ah})$ for the setting with ad hoc sampling, when both are equal we omit the superscript. 

We will show that at every step, the probability that a certain edge gets added to the spanner is the same in both algorithms. We will prove this by induction, under the assumption that the algorithms have led to the same result up to a given point. The base case is easy: here all probabilities are 1, thus both algorithms behave the same. 

Now for the induction step, we assume: 
\begin{itemize}
    \item the first $i-1$ $t$-bundle spanners are created exactly the same $B_j^{(\mathrm{ap})}=B_j^{(\mathrm{ah})}$ for $j<i$,
    \item  the first $m-1$ spanners of the $i$-th $t$-bundle spanner are created the same  $(F^+)_{i,l}^{(\mathrm{ap})}=(F^+)_{i,l}^{(\mathrm{ah})}$,
    \item the first $b-1$ phases of computing the $m$-th spanner have been the same.
\end{itemize}
Moreover, we assume that both algorithms for computing the $m$-th spanner use the same random bits for marking clusters. 

There are in fact multiple induction steps, occurring whenever an edge is chosen to be part of the spanner. These decisions take place in steps 2, 3.1, 3.2, 4.1, 4.2, and 4.3. In each of these steps, the same principle is applied. We will give the proof of the induction procedure at step 2. 

Let $v\in V$ be a vertex in an unmarked cluster. Suppose that \Connect is considering to connect to some neighbor $u$ in an unmarked cluster $X_u$. We have to show that the probability that $u$ is accepted by \Connect with ad hoc sampling, is the same as the probability that it exists in the algorithm with a priori sampling. 

First, suppose that $(u,v)\notin \bigcup_{j,l} F^-_{j,l}$. Let $i'$ be the last $t$-bundle that $(u,v)$ was part of. Then in the ad hoc setting it is accepted by \Connect with probability $1/4^{i-i'}$. In the a priori setting, the edge exists with $1/4$ times the probability it existed in $E^{(\mathrm{ap})}_{i-2}$, resulting in the total probability $1/4^{i-i'}$. 

Now suppose $(u,v)\in F^-_{j,l}$ for some $j,l$. We will show $(u,v)\notin E^{(\mathrm{ap})}_{i-1}$. We proceed by contradiction, so assume $(u,v)\in E^{(\mathrm{ap})}_{i-1}$. Hence also $(u,v)\in E^{(\mathrm{ap})}_{j-1}$. Now we look at the $l$-th spanner of the $j$-th $t$-bundle spanner. Since $(u,v)\in F^-_{j,l}$, we know that two things can be the case. 
\begin{itemize}
    \item When the algorithm with ad hoc sampling called \Connect, this has accepted $u'$ with $w(u',v)>w(u,v)$ or $w(u',v)=w(u,v)$ and $\ID(u')>\ID(u)$. This means that when the algorithm with a priori sampling calls \Connect, it will try $u$ before $u'$ and thus adds $(u,v)$ to $(F^+)^{(\mathrm{ap})}_{j,l}$. This implies $(F^+)^{(\mathrm{ap})}_{j,l}\neq (F^+)^{(\mathrm{ah})}_{j,l}$, a contradiction.
    \item When the algorithm with ad hoc sampling called \Connect, it returned $\bot$. Since $(u,v)$ is an option for the algorithm with a priori sampling. It has at least one option, so will choose some $u'$ (perhaps equal to $u$). Resulting in $(F^+)^{(\mathrm{ap})}_{j,l}\neq (F^+)^{(\mathrm{ah})}_{j,l}$, a contradiction.
\end{itemize}

Similar arguments hold for all other indicated steps, hence by induction, the probabilities that a certain graph $H'$ is equal to the constructed $t$-bundle spanners occurring in the construction of the algorithms are the same. It is left to show that for remaining edges the probability of being added to $E'$ is the same in both algorithms. 

Suppose $e\in E_{\lceil \log m\rceil}^{(\mathrm{ah})}\setminus B_{\lceil \log m\rceil}$. Let $j$ be the index of the last bundle spanner $e$ was part of (possibly zero). 
\begin{itemize}
    \item In the a priori algorithm, the probability of $e$ being added to the next phase is $1/4$ each time. Thus the probability of it lasting until the end is $(1/4)^{\lceil \log m\rceil -j}$.
    \item In the ad hoc algorithm, the probability of $e$ existing gets lowered by a factor $1/4$ each phase, and reset to $1$ if $e$ is part of the bundle spanner. Hence resulting in $(1/4)^{\lceil \log m\rceil -j}$ in the last phase. 
\end{itemize}

Now suppose $e\notin E_{\lceil \log m\rceil}^{(\mathrm{ah})}\setminus B_{\lceil \log m\rceil}$. This means the ad hoc algorithm will not try to add it to $E'$, since it was part of $C_j$ for some $j$. This means in creating the $j$-th bundle spanner, it was considered, but not accepted. As $p\equiv 1$ in the a priori sub procedure of computing the $j$-th bundle spanner, and we know that $B_j^{(\mathrm{ah})}=B_j^{(\mathrm{ap})}$, we can deduce that $e\notin E^{(\mathrm{ap})}_{j-1}$. Thus clearly $e\notin E^{(\mathrm{ap})}_{\lceil \log m\rceil}$.  

We can conclude that if the bundle spanners are created equally, the probability that the algorithms output a specific graph $H$ after the last step will also be equal, which concludes our proof. 
\end{proof}

\begin{theorem}[Theorem 4.1 in \cite{KPPS17}]
\label{thm:SS_apriori}
Given a graph $G=(V,E,w)$ and an error parameter $\epsilon>0$, with high probability, the algorithm \SpectralSparsifyapriori{$V,E,w,\epsilon$} outputs a $(1\pm \epsilon)$-spectral sparsifier $H$ of $G$.
\end{theorem}

\begingroup
\def\thetheorem{\ref{thm:BC_SS}}
\begin{theorem}[Restated]
There exists an algorithm that, given a graph $G=(V,E,w)$ with positive real weights satisfying $||w||_\infty \leq U$ and an error parameter $\epsilon>0$, with high probability outputs a $(1\pm \epsilon)$-spectral sparsifier $H$ of $G$, where $|H|=O\left( n\epsilon^{-2}\log^4 n\right)$. Moreover, we obtain an orientation on $H$ such that with high probability each node has out-degree $O(\log^4(n)/\epsilon^2)$. The algorithm runs in $O\left( \log^5(n)\epsilon^{-2}\log(nU/\epsilon)\right)$ rounds in the Broadcast CONGEST model. 
\end{theorem}
\addtocounter{theorem}{-1}
\endgroup
\begin{proof}
We will prove that the algorithm \SpectralSparsifyapriori{$V,E,w,\epsilon$} satisfies these properties. Correctness follows from Lemma~\ref{lm:SS_alg_eq} together with Theorem~\ref{thm:SS_apriori}. 

The number of edges in the sparsifier is the size of the last $t$-bundle spanner, together with the sampled edges in the remainder. With high probability, this becomes: $$|E(H)|=O\left( n\epsilon^{-2}\log^4 n+ m\cdot \frac{1}{4^{\log n}}\right)=O\left( n\epsilon^{-2}\log^4 n\right).$$

For integer weights bounded by $W$, computing a bundle spanner takes $$O\left(tkn^{1/k}\left(\log n+\log W\right)\right)=O\left( \log^4(n)\epsilon^{-2}\left(\log n+\log W\right)\right)$$ rounds, using that $k=\lceil\log n\rceil$ and $t=400\log^2(n)\epsilon^{-2}$. To run this algorithm on real weights, we multiply all weights by $\Theta(\poly(n)/\epsilon)$ and round to integers. This gives us a graph with integer weights bounded by $W=\Theta(\poly(n)U/\epsilon)$ on which we can compute the bundle spanner in $O(\log^4(n)\epsilon^{-2}\log(nU/\epsilon))$ rounds. Then we scale back the computed bundle spanner with $1/\Theta(\poly(n)/\epsilon)$ to obtain a result that is correct up to error $\epsilon/\poly(n)$, which is clearly sufficient for our purposes. 

Adjusting the probabilities (line 6 through 11) is done internally at each vertex, so it does not affect the number of rounds. A bundle spanner is computed a total of $\lceil\log m\rceil=O(\log n)$ times, thus the algorithm takes $O\left( \log^5(n)\epsilon^{-2}\log(nU/\epsilon)\right)$ rounds. 

Regarding the orientation, we use the orientation from Lemma~\ref{lm:spanner} for the edges that come from a spanner. We orient the remaining edges towards the vertex with the highest ID. 
\end{proof}

\subsection{Laplacian Solving in the Broadcast Congested Clique}\label{sc:BCC_laplacian}
In this section, we restrict ourselves to the Broadcast Congested Clique: by assuming that communication between any two vertices is possible, we can make sure that in the end every vertex knows the entire sparsifier. Since the sparsification algorithm from Theorem~\ref{thm:BC_SS} in fact gives us a way of orienting the edges of the sparsifier such that every vertex has maximum out-degree $O(\log^4(n)/\epsilon^2)$, it can become global knowledge in $O(\log^4(n)/\epsilon^2)$ rounds. However, each edge was explicitly added to the sparsifier in the algorithm, so when run in the BCC, at the end of the algorithm each vertex already knows the entire sparsifier. Now, we can use this spectral sparsifier for Laplacian solving, following Section~\ref{sc:laplacian}.

\begingroup
\def\thetheorem{\ref{thm:laplaciansolveBCC}}
\begin{theorem}[Restated]
There exists an algorithm in the Broadcast Congested Clique model that, given a graph $G=(V,E,w)$, with positive real weights satisfying $||w||_\infty \leq U$ and Laplacian matrix $L_G$, a parameter $\epsilon\in(0,1/2]$, and a vector $b\in \R^n$, outputs a vector $y\in \R^n$ such that $||x-y||_{L_G}\leq \epsilon||x||_{L_G}$, for some $x\in R^n$ satisfying $L_Gx=b$. The algorithm needs $O(\log^5(n)\log(nU))$ preprocessing rounds and takes $O(\log(1/\epsilon)\log(nU/\epsilon))$ rounds for each instance of $(b,\epsilon)$. 
\end{theorem}
\addtocounter{theorem}{-1}
\endgroup
\begin{proof}

The algorithm satisfying these properties is as follows. 
In the preprocessing stage, we find a $(1\pm 1/2)$-spectral sparsifier $H$ for $G$, using \SpectralSparsify{$V,E,w,1/2$}. This takes $O(\log^5(n)\log(nU))$ rounds by Theorem~\ref{thm:BC_SS}.
At the end of this process, $H$ is known to every vertex. Hence any computation with $H$ can be done internally. Also note that multiplying $L_G$ by a vector in the distributed setting only requires each vertex to know the vector values in neighboring vertices and the weights of the edges corresponding to those vertices. Communicating the vector values might need several broadcast rounds, depending on the number of bits necessary to represent the values. Since we aim for error $\epsilon$, $O(\log(nU/\epsilon))$ bits suffice. Hence this takes at most $O(\log(nU/\epsilon))$ rounds. 
Now we apply Corollary~\ref{cor:laplaciansolving} to find~$y$. This takes $O(\log(1/\epsilon))$ iterations of a solve in $L_H$ (done internally at each vertex), a multiplication of $L_G$ by a vector, and a constant number of vector operations (both done in $O(\log(nU/\epsilon))$ rounds). Hence we have a total of $O(\log(1/\epsilon)\log(nU/\epsilon)$ rounds.
\end{proof}

\section{A Linear Program Solver}\label{sc:LP}
In this section, we show how to solve certain linear programs in the Broadcast Congested Clique. The linear programs we consider are distributed over the network in such a way that certain operations with the constraint matrix are easy, these operations are matrix-vector multiplication and certain inversions. 

Our linear program solver consists of an efficient implementation of Lee and Sidford~\cite{LS14, LS19} in the Broadcast Congested Clique, which shows that one can obtain an $\epsilon$-approximation to a linear program using $\Ot(\sqrt{\rm{rank}}\log(1/\epsilon))$ linear system solves. Using our spectral sparsifier based Laplacian solver, we can solve certain linear systems up to the required precision in polylogarithmically many rounds, hence we obtain a $\Ot(\sqrt{\rm{rank}}\log(1/\epsilon))$ round algorithm for solving linear programs that give rise to the correct kind of linear system solves. One such example is in computing minimum cost flows, for which we show how to do this in Section~\ref{sc:BCC_flow}. 

To be precise, let $A\in \R^{m\times n}$, $b\in\R^n$, $c\in\R^m$, $l_i\in \R\cup\{-\infty\}$, and $u_i\in \R\cup\{+\infty\}$ for all $i\in[m]$. The goal is to solve linear programs in the following form
\begin{align*}
{\rm{OPT}}:= \min_{\substack{x\in\R^m : A^Tx=b \\ \forall i\in[m] : l_i\leq x_i\leq u_i}} c^Tx.
\end{align*}
We assume that the interior polytope $\Omega^{\mathrm{o}}:=\{x\in\R^m : A^Tx=b,\ l_i\leq x_i\leq u_i\}$ is non-empty and $\text{dom}(x_i):=\{x: l_i<x_i<u_i\}$ is never the entire real line, i.e., either $l_i\neq-\infty$ or $u_i\neq +\infty$. We then obtain the following theorem.

\begingroup
\def\thetheorem{\ref{thm:BCC_LPSolve}}
\begin{theorem}[Restated]
Let $A \in \R^{m\times n}$ be a constraint matrix with $\rank(A)=n$, let $b\in \R^n$ be a demand vector, and let $c\in \R^m$ be a cost vector. Moreover, let $x_0$ be a given initial point in the feasible region $\Omega^{\mathrm{o}}:=\{x\in\R^m : A^Tx=b,\ l_i\leq x_i\leq u_i\}$. Suppose a Broadcast Congested Clique network consists of $n$ vertices, where each vertex $i$ knows both every entire $j$-th row of $A$ for which $A_{ji}\neq 0$ and knows $(x_0)_j$ if $A_{ji}\neq 0$. Moreover, suppose that for every $y\in \R^n$ and positive diagonal $D\in\R^{m\times m}$ we can compute $(A^TDA)^{-1}y$ up to precision $\poly(1/m)$ in $T(n,m)$ rounds.
Let $U:=\max\{||1/(u-x_0)||_\infty,||1/(x_0-l)||_\infty,||u-l||_\infty,||c||_\infty\}$. Then with high probability the Broadcast Congested Clique algorithm \LPSolve outputs a vector $x\in \Omega^{\mathrm{o}}$ with $c^Tx \leq {\rm{OPT}} + \epsilon$ in $\Ot(\sqrt{n}\log(U/\epsilon)(\log^2(U/\epsilon)+T(n,m)))$ rounds.
\end{theorem}
\addtocounter{theorem}{-1}
\endgroup
As mentioned, the algorithm fulfilling this theorem is an efficient implementation of the LP solver of Lee and Sidford~\cite{LS14,LS19} in the Broadcast Congested Clique. Therefore, we refer to \cite{LS19} for a proof of correctness. We will show how to implement each step in the Broadcast Congested Clique, and bound the running time. Using our Laplacian solver, we can show that we can run most of the Lee-Sidford algorithm directly in the Broadcast Congested Clique. However, there are two subroutines that need adjustment. We need to compute the approximate leverage scores differently (see Section~\ref{sc:barrier_weight_functions}), and we need to adjust the routine for \emph{projections on a mixed norm ball} (See Section~\ref{sc:mixednormball}).

The idea of the algorithm is to use \emph{weighted path finding}, a weighted variant of the standard logarithmic barrier function. In particular we follow a central path reweighted by the $\ell_p$ \emph{Lewis weights}. This means that the barrier function is multiplied by the Lewis weight of the current point. Now each step of the interior point method consists of taking a Newton step and recomputing the weights. 

Throughout this section, we will simplify to the case where $(A^TDA)^{-1}y$ is solved exactly, rather than to precision $\poly(1/m)$. The fact that $\poly(1/m)$ precision suffices is proved by Lee and Sidford in~\cite{LS13}. Throughout the algorithm, vectors will be stored in the natural manner: for $y\in \R^n$ vertex $i$ stores $y_i$, and for $y\in \R^m$ vertex $i$ knows $y_j$ if $A_{ji}\neq 0$. Together with our assumptions on which vertex knows which part of $A$, this means we can perform matrix-vector efficiently. Since we can assume weights and vector values to be at most $O(\poly(n,m)U/\epsilon)$, we obtain this in $\Ot(\log(U/\epsilon))$ rounds.

\subsection{Barriers and Weight Functions}\label{sc:barrier_weight_functions}
Recall that we are minimizing $c^Tx$ for $x \in \Omega^{\mathrm{o}}=\{x\in\R^m : A^Tx=b,\ l_i\leq x_i\leq u_i\}$. To avoid working with these different domains for each $x_i$, Lee and Sidford introduce \emph{1-self-concordant barrier functions} $\phi_i\colon \dom(x_i)\to \R$, which satisfy the following definition.

\begin{definition}\label{def:barrier_function}
	A convex, thrice continuously differentiable function $\phi\colon K\to \R^n$ is a $\nu$-\emph{self-concordant barrier function} for open convex set $K\subseteq \R^n$ if the following three conditions are satisfied
\begin{enumerate}
	\item $\lim_{i\to \infty} \phi(x_i) = \infty$ for all sequences $(x_i)_{i\in \N}$ with $x_i\in K$ converging to the boundary of $K$.
	\item $|D^3\phi(x)[h,h,h]|\leq2|D^2\phi(x)[h,h]|^{3/2}$ for all $x\in K$ and $h\in\R^n$.
	\item $|D\phi(x)[h]|\leq \sqrt{\nu}|D^2\phi(x)[h,h]|^{1/2}$ for all $x\in K$ and $h\in\R^n$.
\end{enumerate}
\end{definition}
In particular, we take $\phi_i$ as follows
\begin{itemize}
	\item If $l_i$ is finite and $u_i=+\infty$, we use a log barrier: $\phi_i(x):=-\log(x-l_i)$.
	\item If $l_i=-\infty$ and $u_i$ is finite, we use a log barrier: $\phi_i(x):=-\log(u_i-x)$.
	\item If $l_i$ and $u_i$ are finite, we use a trigonometric barrier: $\phi_i(x):=-\log \cos(a_ix+b_i)$, where $a_i:=\frac{\pi}{u_i-l_i}$ and $b_i:=-\frac{\pi}{2}\frac{u_i+l_i}{u_i-l_i}$. 
\end{itemize}
One can easily verify that this satisfies Definition~\ref{def:barrier_function}, see \cite{LS19}. Moreover, $\phi(x)$, $\phi'(x)$ and $\phi''(x)$ can all be locally computed in the Broadcast Congested Clique. Barrier functions introduce a central path in the following manner
\begin{align*}
	x_t := \argmin_{A^Tx=b}\left(t\cdot c^Tc+\sum_{i\in[m]}\phi_i(x_i)\right).
\end{align*}
This gives an $\Ot(\sqrt{m}\log(1/\epsilon))$ iteration method for solving the linear program called \emph{path following}~\cite{Renegar88}. Lee and Sidford~\cite{LS14,LS19} show that by weighting the barrier function, this can become an $\Ot(\sqrt{n}\log(1/\epsilon))$ iteration method. Hereto, they look at the \emph{weighted central path}:
\begin{align*}
	x_t := \argmin_{A^Tx=b}\left(t\cdot c^Tc+\sum_{i\in[m]}g_i(x)\phi_i(x_i)\right),
\end{align*}
for some \emph{weight function} $g\colon \Omega^{\mathrm{o}}\to \R^m_{>0}$. Before we introduce the weight function we will be using, let us introduce some shorthand notation. 

\begin{itemize}
    \item For any matrix $M\in \R^{n\times n}$, we let $\diag(M)\in \R^n$ denote the diagonal of $M$, i.e., $\diag(M)_i:=M_{ii}$.
	\item For any vector $x\in \R^n$, we write upper case $X\in\R^{n\times n}$ for the diagonal matrix associated to $x$, i.e., $X_{ii}:=x_i$ and $X_{ij}:=0$ if $i\neq j$.
	\item For $x\in \Omega^{\mathrm{o}}$, we write $A_x:= (\Phi''(x))^{-1/2}A$.
	\item For $h\colon\R^n\to\R^m$ and $x\in\R^n$, we write $J_h(x)\in\R^{m\times n}$ for the Jaccobian of $h$ at $x$, i.e., $[J_h(x)]_{ij}:=\frac{\partial}{\partial x_j}h(x)_i$. 
	\item For positive $w\in \R^n_{>0}$, we let $||\cdot||_w$ the norm defined by $||x||_w^2 = \sum_{i\in[n]} w_ix_i^2$, and we let $||\cdot||_{w+\infty}$ the \emph{mixed norm} defined by $||x||_{w+\infty}=||x||_\infty + C_{\rm{norm}}||x||_w$ for some constant $C_{\rm{norm}}>0$ to be defined later. 
	\item Whenever we apply scalar operation to vectors, these operations are applied coordinate-wise, e.g., for $x,y\in \R^n$ we have $[x/y]_i:=x_i/y_i$, and $[x^{-1}]_i:=x_i^{-1}$.  
\end{itemize}

\begin{definition}
	A differentiable function $g\colon \Omega^{\mathrm{o}}\to \R^m_{>0}$ is a $(c_1,c_{\rm{s}},c_{\rm{k}})$-\emph{weight function} if the following bounds holds for all $x\in  \Omega^{\mathrm{o}}$ and $i\in[m]$:
\begin{itemize}
	\item size bound: $\max\{1,||g(x)||_1\}\leq c_1$;
	\item sensitivity bound: $e_i^TG(x)^{-1}A_x(A_x^TG(x)^{-1}A_x)^{-1}A_x^TG(x)^{-1}e_i\leq c_{\rm{s}}$; 
	\item consistency bound: $||G(x)^{-1}J_g(x)(\Phi''(x))^{-1/2}||_{g(x)+\infty}\leq 1-c_{\rm{k}} <1$. 
\end{itemize}
We denote $C_{\rm{norm}}:=24\sqrt{c_{\rm{s}}}c_{\rm{k}}$.
\end{definition}

In this paper, we use the \emph{regularized Lewis weights}.
\begin{definition}
    For $M\in\R^{m\times n}$ with $\rank(M)=n$, we let $\sigma(M) := \diag(M(M^TM)^{-1}M^T)$ denote the \emph{leverage scores} of $M$. For all $p>0$, we define the \emph{$\ell_p$-Lewis weights} $w_p(M)$ as the unique vector $w\in R^m_{>0}$ such that $w=\sigma(W^{\frac{1}{2}-\frac{1}{p}}M)$, where $w=\diag(W)$. We define the \emph{regularized Lewis weights} as $g(x) := w_p(M_x)+c_0$, for $p=1-\frac{1}{\log(4m)}$ and $c_0 :=\frac{n}{2m}$. 
\end{definition}

We have that the regularized Lewis weight function $g$ is a $(c_1,c_{\rm{s}},c_{\rm{k}})$-weight function with $c_1 \leq \frac{3}{2}n$, $c_{\rm{s}}\leq4$, and $c_{\rm{k}}\leq2\log(4m)$ \cite{LS19}. Computing exact Lewis weights is hard, but we will show that we can compute a sufficient approximation efficiently.

Let us start with the leverage scores. Computing $M(M^TM)^{-1}M^T$ to determine its diagonal is expensive, as are any other known techniques of computing $\sigma(M)$ exactly. However, approximating them is significantly less computationally heavy. As shown in \cite{SS11, DMMW12}, we can reduce the computation to solving a polylogarithmic number of regression problems (see also~\cite{Mahoney11,LMP13,Woodruff14,CLM+15}). Namely by noting that $\sigma(M)_i = ||M(M^TM)^{-1}M^T e_i||_2^2$, and using the Johnson-Lindenstrauss lemma, which states that this norm is preserved approximately under projections onto certain low dimensional subspaces. In previous work \cite{SS11,LS19}, this was done by randomly sampling subspaces. A common approach is to sample polylogarithmically many vectors in $\R^m$ according to some (simple) distribution. However, this is problematic for the Broadcast Congested Clique: it is unclear how one vertex can sample the value of an edge and efficiently communicate this to its corresponding neighbor. Therefore, we use a different variant of the Johnson-Lindenstrauss lemma, by Kane and Nelson~\cite{KN12}, that requires significantly fewer random bits. 
\begin{theorem}[\cite{KN12}]\label{thm:JL}
For any integer $m>0$, and any $\eta>0$, $\delta<1/2$, there exists a family $\mathcal{Q}$ of $k\times m$ matrices for $k= \Theta(\eta^{-2}\log(1/\delta))$ such that for any $x\in \R^m$, 
\[\P_{Q\in\mathcal{Q}}[ (1-\eta)||x||_2 \leq ||Qx||_2 \leq (1+\eta)||x||_2] \geq 1-\delta,\]
where $Q\in \mathcal{Q}$ can be sampled with $O(\log(1/\delta)\log(m))$ uniform random bits. 
\end{theorem}

The following algorithm uses this theorem to compute $\sigma^{(\rm{apx})}$ such that $(1-\eta)\sigma(M)_i \leq \sigma^{(\rm{apx})}_i \leq (1+\eta)\sigma(M)_i$, for all $i\in[m]$. 

\begin{algorithm}[H]
\SetAlgoLined \caption{\textsc{ComputeLeverageScores}($M,\eta$)}\label{alg:leverage_scores}
Set $k=\Theta(\log(m)/\eta^2)$.\\
Broadcast vertex IDs to determine the vertex with the highest ID; declare this vertex the the leader.\\
The leader samples $\Theta(\log^2(m))$ random bits and broadcasts them.\\
Each vertex constructs $Q\in \R^{k\times m}$ from Theorem~\ref{thm:JL} internally, using the random bits sampled by the leader. \\
Compute $p^{(j)}= M(M^TM)^{-1} M^T Q^{(j)}$.\\
\Return{$\sum_{j=1}^{k}\left(p^{(j)}\right)^2$}.
\end{algorithm}

\begin{lemma}\label{lm:leverage_scores}
For any $\eta>0$, with probability at least $1-1/m^{O(1)}$ the algorithm \ComputeLeverageScores{$M,\eta$} computes $\sigma^{\rm{apx}}(M)$ such that 
\[ (1-\eta)\sigma(M)_i\leq \sigma^{\rm{apx}}(M)_i \leq (1+\eta)\sigma(M)_i\]
for all $i\in[m]$. If $M=WA$, for some diagonal $W\in \R^{m\times m}$, it runs in $\Ot(\eta^{-2}(\log(u/\epsilon)+T(n,m)))$ rounds. 
\end{lemma}
\begin{proof}
Note that for any $k\times m$ matrix $Q$ and symmetric $m\times m$ matrix $X$ we have
\begin{align*}
	||QXe_i||_2  = \sum_{j=1}^k (QX)_{ji}^2 = \sum_{j=1}^k (XQ^T)_{ij}^2 = \sum_{j=1}^k \left((XQ^{(j)})_i\right)^2.
\end{align*}
Note that $M(M^TM)^{-1}M^T$ is a symmetric $m\times m$ matrix and $\sigma^{\rm{apx}}(M)_i = ||QM(M^TM)^{-1}M^T e_i||_2^2$, hence $\sigma^{\rm{apx}}(M)= \sum_{j=1}^k \left((M(M^TM)^{-1}M^TQ^{(j)})\right)^2$. By Theorem~\ref{thm:JL}, we have 
\[ (1-\tilde \eta)||M(M^TM)^{-1}M^T e_i||_2 \leq ||QM(M^TM)^{-1}M^T e_i||_2 \leq (1+\tilde\eta)||M(M^TM)^{-1}M^T e_i||_2 \]
with probability at least $1-1/m^{O(1)}$ for our random $Q\in\mathcal{Q}\subseteq \R^{k\times m}$ with $k=\Theta(\tilde\eta^{-2}\log(m))$, constructed from the $\Theta(\log(m)^2)$ random bits sampled by the leader. Using that $\sigma^{\rm{apx}}(M)_i = ||QM(M^TM)^{-1}M^T e_i||_2^2$ and $\sigma(M)_i = ||M(M^TM)^{-1}M^T e_i||_2^2$, we obtain 
\[ (1-\tilde\eta)^2\sigma(M)_i\leq \sigma^{\rm{apx}}(M)_i \leq (1+\tilde\eta)^2\sigma(M)_i.\]
Now setting $\tilde\eta=\eta/4$ gives $1-\eta\leq (1-\tilde\eta)^2$ and $(1+\tilde\eta)^2\leq 1+\eta$, hence we obtain
\[ (1-\eta)\sigma(M)_i\leq \sigma^{\rm{apx}}(M)_i \leq (1+\eta)\sigma(M)_i.\]
This means that we have $k=\Theta(\log(m)/\tilde\eta^2)=\Theta(\log(m)/\eta^2)$. 

For the running time, note that for $j=1, \dots, k$ we need to multiply $M^T$ by a vector, solve a linear system in $M^TM$, and multiply $M$ by a vector. Since $M=WA$, each of these steps can be done in either $\Ot(\log(U/\epsilon))$ rounds or $T(n,m)$ rounds by assumption, giving a total running time of $\Ot(k(\log(U/\epsilon)+T(n,m)))=\Ot(\eta^{-2}(\log(U/\epsilon)+T(n,m)))$ rounds.
\end{proof}

Note that this algorithm is randomized. It is actually the only randomized part of the linear program solver itself and the bottleneck for making it deterministic. However, in the Broadcast Congested Clique our Laplacian solver (Theorem~\ref{thm:laplaciansolveBCC}) is also randomized. On top of that, the algorithm for computing minimum cost maximum flow of Section~\ref{thm:mincostflow_BCC} has an auxiliary randomized component. 

\begin{lemma}\label{lm:weights}
Let $M\in \R^{m\times n}$ with $\rank(M)=n$ be a matrix. For all $\eta \in (0,1)$, $p\in [1-1/\log(4m),2]$, and $w^{(0)}\in \R^m_{>0}$ with $||w_{(0)}^{-1}(w_p(M)-w^{(0)})||_\infty \leq 2^{-20}p^2(4-p)$, the algorithm \ComputeApxWeights{$M,p,w^{(0)},\eta$} returns $w$ such that with high probability $||w_{p}(M)^{-1}(w_p(M)-w||_\infty \leq \eta$. If $M=WA$, for some diagonal matrix $W\in\R^{m \times m}$, the algorithm runs in 
$$\Ot\left(\frac{\log(1/\eta)}{\eta^2}(\log(U/\epsilon)+T(n,m))\right)$$ 
rounds. Further, without given $w^{(0)}$, for $p=1-1/\log(4m)$, the algorithm \ComputeInitialWeights{$p_{\rm{target}},\eta$} returns $w$ such that with high probability $||w_{p_{\rm{target}}}(A)^{-1}(w_{p_{\rm{target}}}(A)-w)||_\infty \leq \eta$ in $$\Ot\left(\left(\sqrt{n}+ \frac{\log(1/\eta)}{\eta^2}\right)(\log(U/\epsilon)+T(n,m))\right)$$ rounds.
\end{lemma}
\begin{proof}
For correctness of the algorithms, we refer to \cite{LS19}. Instead we focus on the implementation in the Broadcast Congested Clique. 

\ComputeApxWeights{$M,p,w^{(0)},\eta$} consists of 
$$T= \left\lceil 80\left(p+\frac{1}{p}\right)\log\left(\frac{pn}{32\eta}\right)\right\rceil = \Ot\left(\left(p+\frac{1}{p}\right)\log(p/\eta)\right)$$
iterations, where in each iteration we call \ComputeLeverageScores with approximation precision $\frac{(4-p)\eta}{512}$, and we compute the coordinate wise median of three vectors. The latter can be done internally, while the former takes $\Ot(((4-p)\eta)^{-2}(\log(U/\epsilon)+T(n,m)))$ rounds, by Lemma~\ref{lm:leverage_scores}. We obtain total number of rounds $$\Ot\left(\left(p+\frac{1}{p}\right)\log(p/\eta)/((4-p)\eta)^2(\log(U/\epsilon)+T(n,m))\right).$$ Using that $p\in[1-1/\log(4m),2]$ this simplifies to $\Ot\left(\frac{\log(1/\eta)}{\eta^2}(\log(U/\epsilon)+T(n,m))\right)$. 

The while loop of \ComputeInitialWeights{$M,p_{\rm{target}},\eta$} terminates after $O(\sqrt{n}(p_{\rm{target}}+\frac{1}{p_{\rm{target}}})\cdot \log(m/n))$ iterations \cite{LS19}, which simplifies to $\Ot(\sqrt{n})$ for $p_{\rm{target}}=1-1/\log(4m)$. Each iteration of the while loop consists of some internal computations and a call to \ComputeApxWeights, these calls have varying values for $p$ and $\eta$, where $p$ starts out as $2$ and gradually changes to $p_{\rm{target}}$, so we always have $p\in [1-1/\log(4m),2]$. In these calls, the allowed error is $\frac{p^2(4-p)}{2^{22}}$, which simplifies to $\Omega(1)$ when $p\in[1-1/\log(4m),2]$. This means we obtain total number of rounds for the while loop of $\Ot\left(\sqrt{n}\cdot (\log(U/\epsilon)+T(n,m))\right)$ rounds. Then there remains one call to \ComputeApxWeights with error $\eta$, which takes $$\Ot\left(\frac{\log(1/\eta)}{\eta^2}(\log(U/\epsilon)+T(n,m))\right)$$ rounds. We obtain a total of $\Ot\left(\left(\sqrt{n}+ \frac{\log(1/\eta)}{\eta^2}\right)(\log(U/\epsilon)+T(n,m))\right)$ rounds.
\end{proof}

\begin{algorithm}[H]
\SetAlgoLined \caption{\textsc{ComputeApxWeights}($M,p,w^{(0)},\eta$)}
$L=\max\{4,\frac{8}{p}\}$, $r=\frac{p^2(4-p)}{2^{20}}$, and $\delta=\frac{(4-p)\eta}{256}$.\\
$T= \left\lceil 80\left(\frac{p}{2}+\frac{2}{p}\right)\log\left(\frac{pn}{32\eta}\right)\right\rceil$. \\
\For{$j=1, \dots, T-1$}{
$\sigma^{(j)}=$\ComputeLeverageScores{$W^{\frac{1}{2}-\frac{1}{p}}_{(j)}M,\delta/2$}.\\
$w^{(j+1)} =$\median{$(1-r)w^{(0)},w^{(j)}-\frac{1}{L}\left(w^{(0)}-\frac{w^{(0)}}{w^{(j)}}\sigma^{(j)}\right),(1+r)w^{(0)}$}.
}
\Return{$w^{(T)}$}.
\end{algorithm}

Lastly, let us turn to the initial weights, which are also computed using the same approximation algorithm, where we gradually transform the all-ones vector into the target weight by repeatedly computing $w_{p'}(A)$, for $p'$ increasingly closer to $p$. 

\begin{algorithm}[H]
\SetAlgoLined \caption{\textsc{ComputeInitialWeights}($p_{\rm{target}},\eta$)}
\SetKwInput{Input}{Input}
$p=2$.\\
$w=12c_{\rm{k}}\mathds{1}$.\\
\While{ $p\neq p_{\rm{target}}$}{
$h=\frac{\min\{2,p\}}{\sqrt{n}\log\frac{me^2}{n}}\cdot r$.\\
$p^{(\rm{new})} =$\median{$p-h,p_{\rm{target}},p+h$}.\\
$w=$\ComputeApxWeights{$A,p^{(\rm{new})},w^{\frac{p^{(\rm{new})}}{p}},\frac{p^2(4-p)}{2^{22}}$}.\\
$p=p^{(\rm{new})}$.
}
\Return{\ComputeApxWeights{$A,p_{\rm{target}},w,\eta$}}.
\end{algorithm}

\subsection{Main Algorithm}
In this section we show how to implement the weighted path finding algorithm, using the weight approximation algorithms of the previous section. The routine \LPSolve is shown in Algorithm~\ref{alg:LPSolve}. 

\begin{algorithm}[H]
\SetAlgoLined \caption{\textsc{LPSolve}($x_0,\epsilon$)} \label{alg:LPSolve}
\SetKwInput{Input}{Input}
\Input{an initial point $x_0$ such that $A^Tx_0=b$.}
$w=$\ComputeInitialWeights{$1-1/\log(4m),\frac{1}{2^{16}\log^3 m}$}$+\frac{n}{2m}$, $d = -w\phi'(x_0)$. \\
$t_1=(2^{27}m^{3/2}U^2\log^4 m)^{-1}$, $t_2=\frac{2m}{\eta}$, $\eta_1=\frac{1}{2^{18}\log^3 m}$, and $\eta_2 =\frac{\epsilon}{8U^2}$.\\
$(x^{(\rm{new})},w^{(\rm{new})}) = $\PathFollowing{$x_0,w,1,t_1,\eta_1,d$}.\\
$(x^{(\rm{final})},w^{(\rm{final})}) = $\PathFollowing{$x^{(\rm{new})},w^{(\rm{new})},t_1,t_2,\eta_2,c$}. \\
\Return{$x^{(\rm{final})}$}.
\end{algorithm}

In this algorithm, the first time we call \PathFollowing, we use it to move the starting point to a more central starting point with respect to the cost vector $c$. The second call to \PathFollowing it to actually solve the problem. First, we take a closer look at the improvement steps by \PathFollowing, see Algorithm~\ref{alg:PathFollowing}. 

\begin{algorithm}[H]
\SetAlgoLined \caption{\textsc{PathFollowing}($x,w,t_{\rm{start}},t_{\rm{end}},\eta,c$)}\label{alg:PathFollowing}
$t=t_{\rm{start}}$, $R=\frac{1}{768c_{\rm{k}}^2\log(36c_1c_{\rm{s}}c_{\rm{k}} m)}$, and $\alpha=\frac{R}{1600\sqrt{n}\log^2 m}$.\\
\While{$t\neq t_{\rm{end}}$}{
$(x,w) = $\CenteringInexact{$x,w,t,c$}. \\
$t \leftarrow $\median{$(1-\alpha)t,t_{\rm{end}},(1+\alpha)t$}.
}
\For{$i=1,\dots,4c_{\rm{k}}\log(\frac{1}{\eta})$}{
$(x,w)= $\CenteringInexact{$x,w,t_{\rm{end}},c$}.
}
\Return{$(x,w)$}.
\end{algorithm}

Here \median{$x,y,z$} simply returns the median of $x$, $y$, and $z$. 
The first loop consists of making a progress step, while the second loop improves centrality. Both are done by calling \CenteringInexact, see Algorithm~\ref{alg:CenteringInexact}. 

\begin{algorithm}[H]
\SetAlgoLined \caption{\textsc{CenteringInexact}($x,w,t,c$)}\label{alg:CenteringInexact}
$R=\frac{1}{768c_{\rm{k}}^2 \log(36c_1c_{\rm{s}}c_{\rm{k}}m)}$, and $\eta=\frac{1}{2c_{\rm{k}}}$.\\
$\delta = \left|\left|P_{x,w}\left(\frac{tc+w\phi'(x)}{w\sqrt{\phi''(x)}}\right)\right|\right|_{w+\infty}$ \tcp{where $P_{x,w}:= I-W^{-1}A_x(A_x^TW^{-1}A_x)^{-1}A_x^T$.} 
$x^{(\rm{new})}= x- \frac{1}{\sqrt{\phi''(x)}} P_{x,w}\left(\frac{tc-w\phi'(x)}{w\sqrt{\phi''(x)}}\right)$.\\
$z =  \log\left(\ComputeApxWeights{$A_{x^{(\rm{new})}},1-1/\log(4m),w,e^R-1$}\right)$.\\
$u = \left(1-\frac{6}{7c_{\rm{k}}}\right)\delta\cdot$\ProjectMixedBall{$-\nabla\Phi\frac{\eta}{12R}(z-\log(w)),C_{\rm{norm}}\sqrt{w}$}.\\
$w^{(\rm{new})}=\exp(\log(w)+u)$.\\
\Return{$\left(x^{(\rm{new})},w^{(\rm{new})}\right)$}.
\end{algorithm}
We present the subroutine \ProjectMixedBall in Section~\ref{sc:mixednormball}. 

The algorithm \CenteringInexact shows how to make a Newton step on $x$ and change the weights $w$ accordingly. 

\begin{lemma}\label{lm:centering}
    The algorithm \CenteringInexact{$x,w,t,c$} runs in $\Ot(\log^2(U/\epsilon)+T(n,m))$ rounds. 
\end{lemma}
\begin{proof}
     In line 2, we first need to compute a vector $y:= \frac{tc+w\phi'(x)}{w\sqrt{\phi''(x)}}$, which can be done internally at each vertex. Next we compute 
\begin{align*}
P_{x,w}y &= y-W^{-1}A_x(A_x^TW^{-1}A_x)^{-1} A_x^Ty \\
&=y-W^{-1}\Phi''(x)^{-1/2}A(A^T\Phi''(x)^{-1/2}W^{-1}\Phi''(x)^{-1/2}A)^{-1} A^T\Phi''(x)^{-1/2}y,
\end{align*}
which we can split up into matrix-vector multiplications, and a linear system solve for $A^TDA$, with $D= \Phi''(x)^{-1/2}W^{-1}\Phi''(x)^{-1/2}$. Matrix-vector multiplications take $\Ot(\log(U/\epsilon))$ rounds, and the linear system solve takes $T(n,m)$ rounds by assumption. Next we need to compute $z$, by calling \ComputeApxWeights with with precision parameter $\eta= e^R-1\geq R = \Omega(1/\log^5 m)$. By Lemma~\ref{lm:weights}, this takes $\Ot(\log(U/\epsilon)+T(n,m))$ rounds. 
Next we call \ProjectMixedBall, which runs in $\Ot(\log^2(U/\epsilon))$ rounds (see Lemma~\ref{lm:mixed_norm_ball}), and we perform some vector operations, which can be done internally. We obtain a total of $\Ot(\log^2(U/\epsilon)+T(n,m))$ rounds.
\end{proof}

Now we use this to analyze the running time of \PathFollowing. 
\begin{lemma}\label{lm:path_following}
    The algorithm \PathFollowing{$x,w,t_{\rm{start}},t_{\rm{end}},\eta,c$} runs in $$\Ot(\sqrt{n}(|\log(t_{\rm{end}}/t_{\rm{start}})|+\log(1/\eta))(\log^2(U/\epsilon)+T(n,m)))$$ rounds. 
\end{lemma}
\begin{proof}
    The while loop of \PathFollowing{$x,w,t_{\rm{start}},t_{\rm{end}},\eta,c$} terminates after at most $$\Ot(\sqrt{n}(|\log(t_{\rm{end}}/t_{\rm{start}})|+\log(1/\eta)))$$ iterations, see \cite{LS19} for a proof. Each iterations consists of some internal computations and a call to \CenteringInexact, which takes $\Ot(\log^2(U/\epsilon)+T(n,m))$ rounds by Lemma~\ref{lm:centering}. The for loop consists of $4c_{\rm{k}}\log(1/\eta)=O(\log(m)\log(1/\eta))$ calls to \CenteringInexact, which is is dominated by the aforementioned while loop. We obtain total a total number of  $$\Ot(\sqrt{n}(\log^2(U/\epsilon)+T(n,m)))$$ rounds.
\end{proof}

Lastly, we consider the complete algorithm. 
\begin{lemma}
    The algorithm \LPSolve{$x_0,\epsilon$} runs in $\Ot(\sqrt{n}\log(U/\epsilon)(\log^2(U/\epsilon)+T(n,m)))$ rounds. 
\end{lemma}
\begin{proof}
    The algorithm consists of three parts:
    \begin{enumerate}
        \item A call to \ComputeInitialWeights{$1-1/\log(4m),\frac{1}{2^{16}\log^3 m}$}. By Lemma~\ref{lm:weights}, this takes $$\Ot((\sqrt{n}+\log(2^{16}\log^3 m)(2^{16}\log^3 m)^2)(\log(U/\epsilon)+T(n,m)))= \Ot(\sqrt{n}(\log(U/\epsilon)+T(n,m)))$$ rounds.
        \item A call to \PathFollowing{$x_0,w,1,t_1,\eta_1,d$} with $t_1=(2^{27}m^{3/2}U^2\log^4 m)^{-1}$ and $\eta_1=\frac{1}{2^{18}\log^3 m}$. By Lemma~\ref{lm:path_following} takes 
        \begin{align*}
            &\Ot(\sqrt{n}(|\log(t_1/1)|+\log(1/\eta_1))(\log^2(U/\epsilon)+T(n,m))) \\
            &= \Ot(\sqrt{n}\log(U)(\log^2(U/\epsilon)+T(n,m)))
        \end{align*}
        rounds.
        
        \item A call to \PathFollowing{$x^{(\rm{new})},w^{(\rm{new})},t_1,t_2,\eta_2,c$} with $t_1=(2^{27}m^{3/2}U^2\log^4 m)^{-1}$, $t_2=\frac{2m}{\eta}$, and $\eta_2 =\frac{\epsilon}{8U^2}$. By Lemma~\ref{lm:path_following} this takes 
        \begin{align*}
            &\Ot(\sqrt{n}(|\log(t_2/t_1)|+\log(1/\eta_2))(\log^2(U/\epsilon)+T(n,m)))\\
            &= \Ot(\sqrt{n}\log(U/\epsilon)(\log^2(U/\epsilon)+T(n,m)))
        \end{align*}
        rounds. 
    \end{enumerate}
    Since the first two operations are dominated by the last, we obtain a total of $$\Ot(\sqrt{n}\log(U/\epsilon)(\log^2(U/\epsilon)+T(n,m)))$$ rounds.
\end{proof}

\subsection{Projection on Mixed Norm Ball in Broadcast Congested Clique}\label{sc:mixednormball}
In this section, we show how to solve the following problem in the Broadcast Congested Clique. Let $a,l\in \R^m$, the goal is to find
\begin{align*}
\argmax_{||x||_2+||l^{-1}x||_\infty\leq 1} a^Tx.
\end{align*}
Hereto, Lee and Sidford~\cite{LS19} initially sort $m$ values and precompute $m$ functions on $a$ and~$l$. For both we have to find an alternative solution, since sorting $m\gg n$ values is difficult and precomputing $m$ functions naively takes $m$ rounds in the Broadcast Congested Clique. We overcome this issue by only sorting implicitly, and doing a binary search, such that we only have to compute logarithmically many functions.

\begin{lemma}\label{lm:mixed_norm_ball}
	Suppose the vectors $a,l\in \R^m$ are distributed over the network such that: 1) for each $i\in[m]$, $a_i$ and $l_i$ are known by exactly one vertex, 2) a vertex knows $a_i$ if and only if it knows $l_i$. Moreover, suppose that $||a||_\infty,||l||_\infty\leq O(\poly(m)U)$. Then there exists an algorithm that finds 
\begin{align*}
\argmax_{||x||_2+||l^{-1}x||_\infty\leq 1} a^Tx
\end{align*}
up to precision $O(1/(\poly(m)\epsilon))$ in $\Ot(\log^2(U/\epsilon)$ rounds in the Broadcast Congested Clique. 
\end{lemma}
\begin{proof}
We describe an algorithm \ProjectMixedBall{$a,l$}, which fulfills the lemma. 

First, we rewrite the problem such that we split the mixed norm into two maximization problems:
\begin{align*}
\max_{||x||_2+||l^{-1}x||_\infty\leq 1} a^Tx &= \max_{0\leq t\leq 1} \left[ \max_{||x||_2\leq 1-t,\ -tl_i\leq x_i \leq tl_i} a^Tx \right]\\
&= \max_{0\leq t\leq 1} (1-t)\left[ \max_{||x||_2\leq 1,\ -\frac{t}{1-t}l_i\leq x_i \leq \frac{t}{1-t}l_i} a^Tx \right].
\end{align*}
For brevity, we write 
\begin{align*}
g(t) := (1-t)\left[ \max_{||x||_2\leq 1,\ -\frac{t}{1-t}l_i\leq x_i \leq \frac{t}{1-t}l_i} a^Tx \right],
\end{align*}
which simplifies the objective to $\max_{0\leq t\leq 1} g(t)$. Now assume that the coordinates are sorted with $|a_i|/l_i$ monotonically decreasing. Later, we will show that we do not have to perform this sorting -- a hard problem in the Broadcast Congested Clique -- explicitly. It can be shown that the vector that attains the maximum in $g(t)$ is $x^{i_t}\in \R^m$, where
\begin{align*}
x^{i_t}_j = \begin{cases} \frac{t}{1-t} \sign(a_j)l_j & \text{if } j\in[i_t] \\ \sqrt{\frac{1-\left(\frac{t}{1-t}\right)^2\sum_{k\in[i_t]}l_k^2}{||a||^2_2-\sum_{k\in [i_t]} a_k^2}a_j} & \text{otherwise. }\end{cases}
\end{align*}
Here we write $i_t$ for the first coordinate $i\in [m]$ such that 
\begin{align*}
\frac{1-\left(\frac{t}{1-t}\right)^2\sum_{k\in[i_t]}l_k^2}{||a||^2_2-\sum_{k\in [i_t]} a_k^2} \leq \frac{\left(\frac{t}{1-t}\right)^2l_i^2}{a_i^2}.
\end{align*}
Note that $i_t \geq i_s$ if $t\leq s$, hence the set of $t$ such that $i_t=j$ is an interval. Now by substitution one can show 
\begin{align*}
g(t) = t\sum_{k\in[i_t]}|a_k||l_k|+\sqrt{(1-t)^2-t^2\sum_{k\in[i_t]} l_k^2}\sqrt{||a||_2^2-\sum_{k\in[i_t]}a_k^2}.
\end{align*}
By looking at its second derivative, one can easily show that that $g(t)$ is a concave function, hence it has a unique maximum. Now we define 
\begin{align*}
g_i(t) = t\sum_{k\in[i]}|a_k||l_k|+\sqrt{(1-t)^2-t^2\sum_{k\in[i]} l_k^2}\sqrt{||a||_2^2-\sum_{k\in[i]}a_k^2}.
\end{align*}
We can rewrite the problem to
\begin{align*}
\max_{0\leq t\leq 1}g(t) &= \max_{0\leq t\leq 1} \max_{i\in[m]} g_i(t)\\
&=  \max_{i\in[m]} \max_{t : i_t = i} g_i(t).
\end{align*}
Suppose for fixed $i\in [m]$ we have calculated $\sum_{k\in[i]}|a_k||l_k|$, $\sum_{k\in[i]} l_k^2$, and $\sum_{k\in[i]} l_k^2$, then internally a vertex can easily calculate $\max_{t : i_t = i} g_i(t)$. 

Moreover, calculating these sums can be done in $O(\log(m)R)$ rounds as follows. For fixed $i\in[m]$, let the vertex $u_i$ denote the vertex that knows $a_i$ and $l_i$. Then $u_i$ computes $|a_i|/l_i$ and broadcasts this value in $R$ rounds. Each vertex $u$ looks at its indices $E_u := \{j: |a_j|/l_j \leq  |a_i|/l_i$\}, which are the entries with $j\leq i$. Then $u$ computes its part of the sum: $\sum_{j\in E_u : } f(j)$, where $f(j)=|a_j||l_j|,a_j^2,l_j^2$. Each vertex $u$ broadcasts the computed sums in $\Ot(\log(U/\epsilon))$ rounds, and the vertex $u_i$ sums the sums to obtain the totals $\sum_{j\in [i]} f(j)$, for each of the three instances of $f$. 

Now we will find the outer maximum by performing a version of binary search over the different $j\in[m]$. Note that we do not actually have the numbers $j$, so we bypass this by doing a binary search over a bigger space. Each vertex broadcasts their minimum and maximum $|a_i|/l_i$, and the least common multiple of of all their denominators of their values $|a_i|/l_i$. From this we find the global minimum and maximum and step size of the binary search, which is the least common multiple of all the least common multiples.  Now we perform binary search on the given range and step size, which has a total of $O(\poly(m)U/\epsilon)$ options. In each iteration, we ask the vertices of the network what their closest value $|a_i|/l_i$ to the aimed value is and use that one, since not all appearing values in the range will be valid. When an $i$ is chosen, we compute $ \max_{t : i_t = i} g_i(t)$ and $ \max_{t : i_t = i+1} g_{i+1}(t)$ to see in which direction the binary search should continue. Note that this is a valid method since the overall function $g(t)$ is concave. The process ends within $\Ot(\log(U/\epsilon))$ iterations, taking a total of $\Ot(\log^2(U/\epsilon)$ rounds. 
\end{proof}

When we apply this lemma in our LP solver, multiple vertices will know the same values $a_i$ and $l_i$, however, each vertex will know which other vertices know $a_i$ and $l_i$, hence we can simply allocate the values $a_i$ and $l_i$ to the vertex with the highest ID.

\section{Minimum Cost Maximum Flow}\label{sc:BCC_flow}
In this section, we apply the linear program solver of the previous section to the minimum cost maximum flow problem. This problem is defined as follows. Let $G=(V,E)$ be a connected directed graph, with capacities $c\colon E \to \Z_{>0}$ and costs $q\colon E\to \Z$\footnote{In this section, we will write $|V|$ and $|E|$, so that we can reserve $n$ and $m$ for the dimensions of the linear program.}. The goal is to compute a maximum flow of minimal cost, see Section~\ref{sc:prelim_flow}. In this section, we prove the following result. 

\begingroup
\def\thetheorem{\ref{thm:mincostflow_BCC}}
\begin{theorem}[Restated]
There exists a Broadcast Congested Clique algorithm that, given a directed graph $G=(V,E)$ with integral costs $q\in \Z^{|E|}$ and capacities $c\in \Z_{>0}^{|E|}$ with $||q||_\infty\leq M$ and $||c||_\infty \leq M$, computes a minimum cost maximum $s$-$t$ flow with high probability in $\Ot(\sqrt{|V|}\log^3 M)$ rounds. 
\end{theorem}
\addtocounter{theorem}{-1}
\endgroup
To prove this, we have to show that the minimum cost maximum flow problem satisfies the conditions of Theorem \ref{thm:BCC_LPSolve}. Clearly, the linear program as presented in Section~\ref{sc:prelim_flow} satisfies this. However, this would incur two problems. The first is that the LP solver computes an \emph{approximate} solution. It is not clear how to efficiently transform this into an exact solution. The second problem is that we need an auxiliary binary search to find the maximum flow. Both problems are solved simultaneously by considering a closely related LP, see Daitch and Spielman~\cite{DS08} and Lee and Sidford~\cite{LS19}. 

We let $B\in\R^{(|V|-1)\times |E|}$ be the edge-vertex incidence matrix where we omit the row for the source $s$. We let our variables consist of $x\in \R^{|E|}, y,z\in \R^{|V|}$ and $F\in \R$. We define the linear program as follows. 

\begin{align*}
\min\ &\tilde{q}^T x+ \lambda(1^Ty+1^Tz)-2n\tilde{M}F\\
\text{\emph{subject to }} & Bx+y-z=Fe_t,\\
&0\leq x_i \leq c_i,\\
&0 \leq y_i\leq 4|V|M,\\
&0 \leq z_i\leq 4|V|M,\\
&0 \leq F\leq 2|V|M,\\
\end{align*}
where $\tilde{M}:=8|E|^2M^3$, $\lambda:= 440 |E|^4 \tilde{M}^2M^3$, and $\tilde{q}$, satisfying $\tilde{q}\leq\tilde{M}$, is defined as follows. For every edge, take a uniformly random number from $\left\{\frac{1}{4|E|^2M^2},\frac{2}{4|E|^2M^2}, \dots, \frac{2|E|M}{4|E|^2M^2}\right\}$, and add this to $q_e$. With probability at least $1/2$, the problem with this cost vector has a unique solution, and this solution is also a valid solution for the original problem~\cite{DS08}\footnote{We can easily boost the success probability from $1/2$ to $1-n^c$ at the cost of a factor $O(c \log n)$, by running the algorithm $O(c\log n)$ times. Each time we can check in one round whether the flow is feasible: each node checks the constraints and broadcasts whether they are satisfied. In the end we take the flow of minimum cost of all flows of maximum value. With probability at least $1-n^c$ this is the minimum cost maximum flow.}. We apply this reduction and scale the problem such that the cost vector is integral again.

It is easy to check that the following is an interior point: $F=|V|M, x=\frac{c}{2}, y=2|V|M\mathds{1}-(B\frac{c}{2})^-+Fe_t, z=2|V|M\mathds{1}+(B\frac{c}{2})^+$, where we denote $a^+$ and $a^-$ for the vectors defined by 
\[ (a^+)_i := \begin{cases} a_i &\text{if }a_i\geq 0; \\ 0 &\text{else.}\end{cases} \hspace{3em}\text{and}\hspace{3em} (a^-)_i := \begin{cases} a_i &\text{if }a_i\leq 0 \\ 0 &\text{else.}\end{cases} \]
respectively. 

A solution to this linear program can be transformed to a solution to the minimum cost maximum flow problem. To be precise, one can find an exact solution to the minimum cost maximum $s$-$t$ flow problem, if we can find a feasible solution to the above LP with cost value within $\frac{1}{12M}$ of the optimum. This solution $x$ is then transformed in two steps: it is made a feasible flow $\tilde{x}$ for the original graph by subtracting the error we may have created by introducing additional variables $y$ and $z$. One can show this is at most $\epsilon:= \frac{1}{40|E|^2\tilde{M}M}$, since the LP is solved up to precision  $\frac{1}{12M}$ \cite{LS19}. We set $\tilde{x}:=(1-\epsilon)x$. This is not yet optimal, but  by integrality of costs and the fact that the min-cost solution is unique, we have that the flow $\tilde{x}_e$ on each edge is at most $1/6$ off from the optimal value \cite{LS19}. We obtain the optimal value simply by rounding $\tilde{x}_e$ to the closest integer. In the Broadcast Congested Clique, multiplication by $(1-\epsilon)$ and rounding can be done internally, so this requires no rounds. 

Next, we show how to actually solve the above LP. We set $A: = [B\ I\ -I\ -e_t]^T$, and use the LP solver of Section~\ref{sc:LP} with this constraint matrix. To be precise, we run the algorithm on a network of $n=|V|-1$ vertices, as $s$ does not need to participate. It is clear the knowledge of $A$ is distributed in the required manner, as local knowledge of the edge-vertex incidence matrix $B$ is known by default.

Next we show that we can solve linear equations in $A^TDA$ in $\Ot(\log(M))$ rounds. Then Theorem~\ref{thm:BCC_LPSolve} solves the LP in $\Ot(\sqrt{n}\log^3 M)$ rounds, where we use that $T(n,m)=\Ot(\log(M))$. 
\begin{lemma}
Let $D\in\R^{(|E|+2|V|-1)\times(|E|+2|V|-1)}$ be any positive diagonal matrix, then there is a BCC algorithm that solves linear equations in $A^TDA$ up to precision $1/m^{O(1)}$ in $\Ot(\log(M))$ rounds.
\end{lemma}
\begin{proof}
We write 
\begin{align*}
D = \begin{pmatrix}
D_1 & 0 & 0 &0\\
0 & D_2 & 0 &0 \\
0 &0 & D_3 & 0\\
0 & 0 & 0 & D_4
\end{pmatrix},
\end{align*}
for $D_1 \in \R^{|E|\times |E|}$, $D_2,D_3\in \R^{(|V|-1)\times(|V|-1)}$ and $D_4 \in \R$ the diagonal submatrices, then we can rewrite $A^T DA$ as
\begin{align*}
A^T DA &= [B\ |\ I\ |\ -I\ |\ -e_t] D [B\ |\ I\ |\ -I\ |\ -e_t]^T\\
&=  [B\ |\ I\ |\ -I\ |\ -e_t]  [BD_1\ |\ D_2\ |\ -D_3\ |\ -e_tD_4]^T\\
&= BD_1B^T+D_2+D_3+e_tD_4e_t^T.
\end{align*}
We will show that we can locally compute $BD_1B^T$. The other three components are either known locally or easy to compute, so we can determine the final $(|V|-1)\times (|V|-1)$ matrix.
\begin{align*}
(BD_1B^T)_{u,v} &= \sum_{e\in E} B_{u,e} (D_1)_{e,e} B^T_{e,v} \\
&=  \sum_{e\in E} B_{u,e} (D_1)_{e,e} B_{v,e}
\end{align*}
Note that if $u\notin e$ or $v\notin e$, then the summand is zero. Now we get a case distinction between $u=v$, and $u\neq v$
\begin{align*}
(BD_1B^T)_{u,v} = \begin{cases} 
-(D_1)_{(u,v),(u,v)}-(D_1)_{(v,u),(v,u)} & \text{if } u\neq v\\
\sum\limits_{e\in E: u\in e} (D_1)_{e,e}  & \text{if } $u=v$,
\end{cases}
\end{align*}
where we write by abuse of notation that $(D_1)_{(u,v),(u,v)}=0$ if $(u,v)\notin E$. From this formula, it is immediate that each vertex can determine its row and column in the matrix, since this has only contributions from edges in which this vertex takes part, hence are known to it. Moreover, the matrix $M:=BD_1B^T+D_2+D_3+e_tD_4e_t^T$ is clearly symmetric and it is diagonally dominant. The latter holds since $BD_1B^T$ is the only summand with non-zero off-diagonal entries, and here the sum of each row/column is zero. 

We conclude that we have to solve a linear system with symmetric diagonally dominant (SDD) $L$, a $(|V|-1)\times(|V|-1)$ matrix $M$ where each vertex knows its row/column. In general, this matrix is not a Laplacian matrix for which we can use our Broadcast Congested Clique Laplacian solver. However, there is a standard reduction from SDD systems to Laplacian systems. This reduction is first presented by Gremban~\cite{Gremban96}, here we use the notation of Kelner et al.~\cite{KOSZ13}.
We get
\begin{itemize}
	\item $M_n$ all the negative off-diagonal entries of $M$;
	\item $M_p$ all the positive off-diagonal entries of $M$;
	\item $C_1$ the diagonal matrix defined by $C_1(u,u) := \sum_{v} |M(u,v)|$;
	\item $C_2:=M-M_n-C_1$. 
\end{itemize}
Note that in our case we have $M_p=0$ and that the contribution of each vertex to each of the other matrices can simply be determined internally. 
Now we get a Laplacian matrix $L$ defined by 
\begin{align*}
L = \begin{pmatrix}
C_1 + C_2/2+M_n & -C_2/2-A_p\\
-D_2/2 - A_p & D_1+D_2/2+A_n 
\end{pmatrix}.
\end{align*}
It is not hard to verify that $L$ is Laplacian, and that, given $y\in\R^{|V|-1}$ , we have that an (approximate) solution $\begin{pmatrix}x_1\\x_2\end{pmatrix}$ to $L\begin{pmatrix}x_1\\x_2\end{pmatrix}=\begin{pmatrix}y\\-y\end{pmatrix}$ implies that $x=\frac{x_1-x_2}{2}$ is an (approximate) solution to $Mx=y$. So we turn to solving equations in $L$, which is a Laplacian matrix of a virtual graph~$G'$ on $2(n-1)$ vertices. Note that $L$ is a $2(|V|-1)\times2(|V|-1)$ matrix, which we can simulate on our network of $|V|-1$ vertices ($s$ remains inactive) by letting vertex $i$ simulate row $i$ and row $i+|V|-1$. This works well, as these are exactly the matrix entries that it knows. Now we can run our BCC Laplacian solver (see Theorem~\ref{thm:laplaciansolveBCC}) in $\Ot(\log(M))$ rounds on this graph. We do this by simulating each round by two rounds: in the first one each vertex $i$ sends messages corresponding to virtual vertex $i$, and in the second round it sends messages corresponding to virtual vertex $i+|V|-1$. 
\end{proof}

\input{conclusion}

\subsection*{Acknowledgements}
This work is supported by the Austrian Science Fund (FWF): P 32863-N. This project has received funding from the European Research Council (ERC) under the European Union's Horizon 2020 research and innovation programme (grant agreement No~947702).

\printbibliography[heading=bibintoc] 
\appendix

\section{Baswana-Sen Spanner Algorithm} \label{app:BSalg}
Below we provide a randomized algorithm, computing a $(2k-1)$-spanner. This algorithm is introduced by Baswana and Sen~\cite{BS07}, we give a rephrased version from Becker et al.~\cite{BeckerFKL21}. 

\begin{enumerate}
    \item Initially, each vertex is a singleton cluster: $R_1 := \{\{v\}\ |\ v\in V\}$.
    \item  For $i=1,\dots, k-1$ do:
    \begin{enumerate}
        \item Each cluster from $R_i$ is marked independently with probability $n^{-1/k}$. $R_{i+1}$ is defined to be the set of clusters marked in phase $i$.
        \item If $v$ is a vertex in an unmarked cluster:
        \begin{enumerate}
            \item Define $Q_v$ to be the set of edges that consists of the lightest edge from $v$ to each cluster in $R_i$ it is adjacent to.
            \item If $v$ is not adjacent to any marked cluster, all edges in $Q_v$ are added to the spanner.
            \item Otherwise, let $u$ be the closest neighbor of $v$ in a marked cluster. In this case, $v$ adds to the spanner the edge $\{v,u\}$ and all edges $\{v,w\} \in Q_v$ with $w(v,w) < w(v,u)$ (break ties by neighbor identifiers). Also, let $X$ be the cluster of $u$. Then $X := X\cup \{v\}$, i.e., $v$ \emph{joins} the cluster of $u$.
        \end{enumerate}
    \end{enumerate}
    \item Each $v\in V$ adds, for each $X \in R_k$ it is adjacent to, the lightest edge connecting it to $X$ to the spanner.
\end{enumerate}

\end{document}

%% file: commands.tex
\newcommand*\samethanks[1][\value{footnote}]{\footnotemark[#1]}

\allowdisplaybreaks[1]

\makeatletter
\g@addto@macro\bfseries{\boldmath}
\makeatother

\makeatletter
\g@addto@macro\mdseries{\unboldmath}
\g@addto@macro\normalfont{\unboldmath}
\g@addto@macro\rmfamily{\unboldmath}
\g@addto@macro\upshape{\unboldmath}
\makeatother

\DeclareCiteCommand{\citem}
    {}
    {\mkbibbrackets{\bibhyperref{\usebibmacro{postnote}}}}
    {\multicitedelim}
    {}

\renewcommand*{\multicitedelim}{\addcomma\space}

\AtEveryCitekey{\clearfield{note}}

\makeatletter
\AtBeginDocument{%
    \newlength{\temp@x}%
    \newlength{\temp@y}%
    \newlength{\temp@w}%
    \newlength{\temp@h}%
    \def\my@coords#1#2#3#4{%
      \setlength{\temp@x}{#1}%
      \setlength{\temp@y}{#2}%
      \setlength{\temp@w}{#3}%
      \setlength{\temp@h}{#4}%
      \adjustlengths{}%
      \my@pdfliteral{\strip@pt\temp@x\space\strip@pt\temp@y\space\strip@pt\temp@w\space\strip@pt\temp@h\space re}}%
    \ifpdf
      \typeout{In PDF mode}%
      \def\my@pdfliteral#1{\pdfliteral page{#1}}
      \def\adjustlengths{}%
    \fi
    \ifxetex
      \def\my@pdfliteral #1{}
      \def\adjustlengths{\setlength{\temp@h}{-\temp@h}\addtolength{\temp@y}{1in}\addtolength{\temp@x}{-1in}}%
    \fi%
    \def\Hy@colorlink#1{%
      \begingroup
        \ifHy@ocgcolorlinks
          \def\Hy@ocgcolor{#1}%
          \my@pdfliteral{q}%
          \my@pdfliteral{7 Tr}
        \else
          \HyColor@UseColor#1%
        \fi
    }%
    \def\Hy@endcolorlink{%
      \ifHy@ocgcolorlinks%
        \my@pdfliteral{/OC/OCPrint BDC}%
        \my@coords{0pt}{0pt}{\pdfpagewidth}{\pdfpageheight}%
        \my@pdfliteral{F}
        \my@pdfliteral{EMC/OC/OCView BDC}%
        \begingroup%
          \expandafter\HyColor@UseColor\Hy@ocgcolor%
          \my@coords{0pt}{0pt}{\pdfpagewidth}{\pdfpageheight}%
          \my@pdfliteral{F}
        \endgroup%
        \my@pdfliteral{EMC}%
        \my@pdfliteral{0 Tr}
        \my@pdfliteral{Q}%
      \fi
      \endgroup
    }%
}
\makeatother

%% file: introduction.tex
\section{Introduction}

In this paper, we study algorithms for the \emph{Broadcast Congested Clique (BCC)} model~\cite{DKO12}.
In this model, the (problem-specific) input is distributed among several processors and the goal is that at the end of the computation each processor knows the output or at least the share of the output relevant to it.
The computation proceeds in rounds and in each round each processor can send one message to all other processors.
We can also view the communication as happening via a shared blackboard to which each processor may write (in the sense of appending) at most one message per round.
The main metric in designing and analyzing algorithms for the Broadcast Congested Clique is the number of rounds performed by the algorithm.

A typical way of for example distributing an $ n \times n $ input matrix among $ n $ processors would be that initially processor~$ i $ only knows row~$ i $ of the matrix.
In many graph problems, this input matrix is the adjacency matrix of the graph.
If communication with other processors is only possible along the edges of this graph, then the resulting model is often called the Broadcast CONGEST model~\cite{Lynch96}.
Note that the unicast versions of these models, in which each processor may send a different message to each (neighboring) processor, are known as the Congested Clique~\cite{LPSPP05} and the CONGEST model~\cite{Peleg00}, respectively.

In this paper, we bring the main tools of the so-called \emph{Laplacian paradigm} to the BCC model.
In a seminal paper, Spielman and Teng developed an algorithm for approximately solving linear systems of equations with a Laplacian coefficient matrix in a near-linear number of operations~\cite{ST14}.
The Laplacian paradigm~\cite{Teng10} refers to exploring the applications of this fast primitive in algorithm design.
In a broader sense, this paradigm is also understood as the more general idea of employing linear algebra methods from continuous optimization outside of their traditional domains.
Using such methods is very natural in distributed models because a matrix-vector multiplication can be carried out in a single round if each processor stores one coordinate of the vector.
In recent years, this methodology has been successfully employed in the CONGEST model~\cite{GKK+15,BeckerFKL21} and in particular, solvers for Laplacian systems with near-optimal round complexity have been developed for the CONGEST model -- in networks with arbitrary topology~\cite{FGLP+20} and in bounded-treewidth graphs~\cite{AGL21} -- and for the HYBRID model~\cite{AGL21}.
In this paper, we switch the focus to the BCC model and show that it allows a faster implementation of the basic Laplacian primitive.

What further makes the BCC model intriguing is that -- in contrast to the Congested Clique -- for several problems no tailored BCC algorithms are known that are significantly faster than low-diameter versions of (Broadcast) CONGEST model algorithms.
Consider, for example, the single-source shortest path problem.
In the (Broadcast) CONGEST model, the fastest known algorithm takes $ \Ot (\sqrt{n} D^{1/4} + D) $ rounds~\cite{ChechikM20}, where $ D $ is the diameter of the underlying (unweighted) communication network.\footnote{ Throughout the introductory part of this paper we often assume that all weights of graphs and entries of matrices are polynomially bounded to simplify some statements of running time bounds.}
In the BCC model, the state of the art for this problems is $ \Ot (\sqrt{n}) $ rounds~\cite{Nanongkai14}, which essentially is not more efficient than the special case $ D = 1 $ of the Broadcast CONGEST model.
In the Congested Clique model however, $ \sqrt{n} $ is not a barrier for this problem as it can be solved in $ \Ot (n^{1/6}) $ rounds~\cite{CDKL21} on undirected graphs.
A similar classification can be made for directed graphs~\cite{ForsterN18,Censor-HillelKK19}.
This naturally leads to the question whether BCC algorithms can be developed that are faster than their CONGEST model counterparts, since it is not clear which one dominates the other in strength. 

It has recently been shown that in the CONGEST model, the maximum flow problem as well as the unit-capacity minimum cost flow problem can be solved in $ O (m^{3/7 + o(1)} (\sqrt{n} D^{1/4} + D)) $ rounds~\cite{FGLP+20}, where $ m $ denotes the number of edges of the input graph; note that this round complexity can only be sublinear in $ n $ for sparse graphs.

\paragraph{Our contributions.}
Our main result is an algorithm that solves the minimum cost flow problem\footnote{Note that in contrast to the algorithm of Forster et al.~\cite{FGLP+20}, we do not need to assume unit capacities.} (which generalizes both the single-source shortest path problem and the maximum flow problem) in $ \Ot (\sqrt{n}) $ rounds in the BCC model, which in particular is sublinear for any graph density and matches the currently known upper bounds for the single-source shortest paths problem.

\begin{theorem}\label{thm:mincostflow_BCC}
There exists a Broadcast Congested Clique algorithm that, given a directed graph $G=(V,E)$ with integral costs $q\in \Z^{m}$ and capacities $c\in \Z_{>0}^{m}$ with $||q||_\infty\leq M$ and $||c||_\infty \leq M$, computes a minimum cost maximum $s$-$t$ flow with high probability in $\Ot(\sqrt{n}\log^3 M)$ rounds. 
\end{theorem}

In obtaining this result, we develop machinery of the Laplacian paradigm that might be of independent interest.
The first such tool is an algorithm for computing a spectral sparsifier in the Broadcast CONGEST model.
\begin{theorem}\label{thm:BC_SS}
There exists an algorithm that, given a graph $G=(V,E,w)$ with positive real weights satisfying $||w||_\infty \leq U$ and an error parameter $\epsilon>0$, with high probability outputs a $(1\pm \epsilon)$-spectral sparsifier $H$ of $G$, where $|H|=O\left( n\epsilon^{-2}\log^4 n\right)$. Moreover, we obtain an orientation on $H$ such that with high probability each edge has out-degree $O(\log^4(n)/\epsilon^2)$. The algorithm runs in $O\left( \log^5(n)\epsilon^{-2}\log(nU/\epsilon)\right)$ rounds in the Broadcast CONGEST model. 
\end{theorem}

At a high level, our sparsifier algorithm is a modification of the CONGEST-model algorithm of Koutis and Xu~\cite{KX16}; essentially, uniform edge sampling is trivial in the CONGEST model, but challenging in the Broadcast CONGEST model.
Note that the sparsifier algorithm of Koutis and Xu being restricted to the CONGEST model is a major obstacle for implementing the CONGEST-model Laplacian solver of Forster et al.~\cite{FGLP+20} also in the Broadcast CONGEST model.

Making the sparsifier known to every processor leads to a simple residual-correction algorithm for solving systems of linear equations with a Laplacian coefficient matrix up to high precision in the BCC model. Note that there is reduction~\cite{Gremban96} from solving linear equations with symmetric diagonally dominant (SDD) coefficient matrices to solving linear equations with Laplacian coefficient matrices, which also applies in the Broadcast Congested Clique.
\begin{theorem}\label{thm:laplaciansolveBCC}
There exists an algorithm in the Broadcast Congested Clique model that, given a graph $G=(V,E,w)$, with positive real weights satisfying $||w||_\infty \leq U$ and Laplacian matrix $L_G$, a parameter $\epsilon\in(0,1/2]$, and a vector $b\in \R^n$, outputs a vector $y\in \R^n$ such that $||x-y||_{L_G}\leq \epsilon||x||_{L_G}$, for some $x\in R^n$ satisfying $L_Gx=b$. The algorithm needs $O(\log^5(n)\log(nU))$ preprocessing rounds and takes $O(\log(1/\epsilon)\log(nU/\epsilon))$ rounds for each instance of $(b,\epsilon)$. 
\end{theorem}

Finally, we show how to implement the algorithm of Lee and Sidford~\cite{LS14}\footnote{Note that in the more technical parts of our paper we explicitly refer to the arXiv preprints~\cite{LS13} and~\cite{LS19} instead of the conference version~\cite{LS14}.} for solving linear programs up to small additive error in $ \Ot (\sqrt{\text{rank}}) $ iterations in the BCC model.
Here, the rank refers to the constraint matrix of the LP and in each iteration a linear system needs to be solved.
If the constraint matrix has a special structure -- which is the case for the LP formulation of the minimum cost flow problem -- then a high-precision Laplacian solver can be employed for this task.
\begin{theorem}\label{thm:BCC_LPSolve}
Let $A \in \R^{m\times n}$ be a constraint matrix with $\rank(A)=n$, let $b\in \R^n$ be a demand vector, and let $c\in \R^m$ be a cost vector. Moreover, let $x_0$ be a given initial point in the feasible region $\Omega^{\mathrm{o}}:=\{x\in\R^m : A^Tx=b,\ l_i\leq x_i\leq u_i\}$. Suppose a Broadcast Congested Clique network consists of $n$ vertices, where each vertex $i$ knows both every entire $j$-th row of $A$ for which $A_{ji}\neq 0$ and knows $(x_0)_j$ if $A_{ji}\neq 0$. Moreover, suppose that for every $y\in \R^n$ and positive diagonal $D\in\R^{m\times m}$ we can compute $(A^TDA)^{-1}y$ up to precision $\poly(1/m)$ in $T(n,m)$ rounds.
Let $U:=\max\{||1/(u-x_0)||_\infty,||1/(x_0-l)||_\infty,||u-l||_\infty,||c||_\infty\}$. Then with high probability the Broadcast Congested Clique algorithm \LPSolve outputs a vector $x\in \Omega^{\mathrm{o}}$ with $c^Tx \leq {\rm{OPT}} + \epsilon$ in $\Ot(\sqrt{n}\log(U/\epsilon)(\log^2(U/\epsilon)+T(n,m)))$ rounds.
\end{theorem}
While this approach of solving LPs is inherently parallelizable (as the PRAM depth analysis of Lee and Sidford indicates), several steps pose a challenge for the BCC model and require more than a mere ``translation'' between models. In particular we need to use a different version of the Johnson-Lindenstraus lemma to approximate leverage scores. Further we give a BCC algorithm for projecting vectors on a mixed norm ball. 

As in the approach of Lee and Sidford, our main result on minimum cost maximum flow then follows from plugging a suitable linear programming formulation of the problem into the LP solver.

\paragraph{Overview.}

We provide a visual overview of the results in this paper and how they are interconnected in Figure~\ref{fig:overview}.
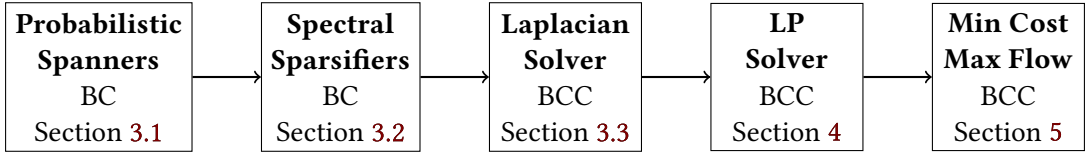
\begin{figure}[h!]
\centering
\begin{tikzpicture}[node distance=.9cm,every text node part/.style={align=center}]
\tikzstyle{alg} = [rectangle, minimum width=2cm, minimum height=1cm, text centered, draw=black]
\node (1) [alg] {\textbf{Probabilistic} \\ \textbf{Spanners}\\ BC\\Section~\ref{sc:prob_spanners}};
\node (2) [alg, right = of 1] {\textbf{Spectral}\\ \textbf{Sparsifiers}\\ BC\\Section~\ref{sc:BC_sparsifier}};
\node (3) [alg, right = of 2] {\textbf{Laplacian}\\ \textbf{Solver}\\ BCC\\Section~\ref{sc:BCC_laplacian}};
\node (4) [alg, right = of 3] {\textbf{LP} \\ \textbf{Solver}\\ BCC\\Section~\ref{sc:LP}};
\node (5) [alg, right = of  4] {\textbf{Min Cost}\\ \textbf{Max Flow}\\ BCC\\Section~\ref{sc:BCC_flow}};

\path[->,thick]
(1) edge (2)
(2) edge (3)
(3) edge (4)
(4) edge (5)
;
\end{tikzpicture}
\caption{An overview of the results in this paper. We denote BC for Broadcast CONGEST and BCC for Broadcast Congested Clique. }
\label{fig:overview}
\end{figure}

To compute spectral sparsifiers in the Broadcast CONGEST model, we follow the setup of Koutis and Xu~\cite{KX16}.
Roughly said, this consists of repeatedly computing spanners and retaining each edge that is not part of a spanner with probability $1/4$.
While this easily allows for an implementation in the CONGEST model (as pointed out by Koutis and Xu), it is not clear how to do this in a broadcast model -- neither the Broadcast CONGEST model, nor the more powerful Broadcast Congested Clique.\footnote{We believe that it would be interesting to explore whether the bounded-independence sampling technique of Doron et al.~\cite{DMVZ20} could also be applied to the algorithmic framework of Koutis and Xu~\cite{KX16}.
Such a sampling method based on a random seed of polylogarithmic size would also significantly simplify an argument in the quantum sparsifier algorithm of Apers and de Wolf~\cite{ApersW20}. Note that in the Broadcast Congested Clique model, a designated vertex could initially sample such a small random seed and communicate it to all other vertices with only a polylogarithmic overhead in round complexity. In the Broadcast CONGEST model, such an approach would however lead to an overhead of $ \Omega(D) $ rounds (the diameter of the underlying communication network), which, as we show, is avoidable.}
A straightforward way to sample an edge would be that one of its endpoints (say the one with the lower ID) decides if it should further exist.
The problem with this approach is that a vertex might be responsible for performing the sampling of a polynomial number of incident edges and the broadcast constraint prevents this vertex from sharing the result with each of the corresponding neighbors.
We overcome this obstacle as follows.
We explicitly maintain the probability that an edge still exists in the current iteration of the sparsifier algorithm of Koutis and Xu.
Every time that an edge should be added to the current iteration's spanner according to the spanner algorithm, one of the endpoints samples whether the edge exists using the maintained probability.
Due to the vertex' subsequent action in the spanner algorithm, the corresponding neighbor can deduce the result of the sampling.
We show that this idea of implicitly learning about the result of the sampling can be implemented by modifying the spanner algorithm of Baswana and Sen~\cite{BS07}.
We present our modification to compute a spanner on a ``probabilistic'' graph (in the sense described above) in Section~\ref{sc:prob_spanners}.
In Section~\ref{sc:BC_sparsifier}, we prove that this can be plugged into the framework of Koutis and Xu to compute a spectral sparsifier in the Broadcast CONGEST model.
Subsequently, we show in Section~\ref{sc:BCC_laplacian} that the spectral sparsifier can be used for Laplacian solving with standard techniques.

In Section~\ref{sc:LP}, we present our LP solver.
Given a linear program of the form\footnote{Following Lee and Sidford, we write $A^Tx=b$ instead of the more common $Ax=b$ for the linear program, since this means that $n$ corresponds with the number of vertices and $m$ with the number of edges in LP formulations of flow problems.}
\[ \min_{x\in R : A^Tx=b} c^Tx, \]
for some constraint matrix $A\in\R^{m\times n}$ and some convex region $R \subseteq\R^m$, Lee and Sidford~\cite{LS14,LS19} show how to find an $\epsilon$-approximate solution in $\Ot(\sqrt{\rank{(A)}}\log(1/\epsilon))$ time. An implementation of this algorithm in the Broadcast Congested Clique is rather technical and needs new subroutines, the main one being our Laplacian solver.

The algorithm is an interior point method that uses weighted path finding to make progress. The weights used are the \emph{Lewis weights}, which can be approximated up to sufficient precision using the computation of leverage scores, which are defined as $\sigma(M):=\diag(M(M^TM)^{-1}M^T)$, where in our case $M=DA$, for some diagonal matrix $D$. Computing leverage scores exactly is expensive, hence these too are approximated. This can be done using the observation that $\sigma(M)_i = ||M(M^TM)^{-1}M^T e_i||_2^2$ and the Johnson-Lindenstrauss lemma~\cite{JL84}, which states that there exists a map $Q\in \R^{k\times m}$ such that $ (1-\epsilon)||x||_2\leq ||Qx||_2\leq (1+\epsilon)||x||_2$, for polylogarithmic~$k$. Nowadays, several different (randomized) constructions for $Q$ exist. A common choice in the realm of graph algorithms~\cite{SS11, LS19} is to use Achlioptas' method~\cite{Achlioptas03}, which samples each entry of $Q$ with a binary coin flip. However, this is in practice not feasible in the Broadcast Congested Clique: we would need a coin flip for every edge, which can be performed by one of the endpoints, but cannot be communicated to the other endpoint due to the broadcast constraint. Instead we use the result of Kane and Nelson~\cite{KN12}, that states we need only a polylogarithmic number of random bits in total. These can simply be sampled by one vertex and broadcast to all the other, who then internally construct $Q$. Now if we can multiply both $A$ and $A^T$ by a vector, and solve linear systems involving $A^TDA$, for diagonal $D$, then we can compute these leverage scores efficiently. These demands on $A$ are not unreasonable when we consider graph problems, because in such cases the constraint matrix will adhere to the structure of the graph Laplacian, and hence our Laplacian solver can be applied.

A second challenge in implementing Lee and Sidford's LP solver is a subroutine that computes projections on a mixed norm ball. To be precise: for $a,l\in \R^m$ distributed over the network, the goal is to find
\begin{align*}
\argmax_{||x||_2+||l^{-1}x||_\infty\leq 1} a^Tx.
\end{align*}
We show that we can solve this maximization problem when we know the sums $\sum_{k\in[i]}a_k^2$, $\sum_{k\in[i]}l_k^2$, and $\sum_{k\in[i]}|a_k||l_k|$ for all $i\in[m]$. Computing such a sum for fixed $i$ is feasible in a polylogarithmic number of rounds. Moreover, we show that we do not need to inspect these sums for all $i\in[m]$, but that we can do a binary search, which reduces the total time complexity to polylogarithmic. 

Following Lee and Sidford~\cite{LS14}, we apply the LP solver to an LP formulation of the minimum cost maximum flow problem in Section~\ref{sc:BCC_flow}.
The corresponding constraint matrix $ A $ has $ O (n) $ rows and thus rank $ O(\sqrt{n}) $.
Furthermore, $ (A^TDA) $ (for any diagonal matrix $ D $) is symmetric diagonally dominant and thus $(A^TDA)^{-1}y$ can be approximated to high precision in a polylogarithmic number of rounds with our Laplacian solver. We only need to solve the LP up to precision $\Theta(1/m^{O(1)})$, since we can round the approximate solution to an exact solution. Hence, the minimum maximum cost flow LP can be solved in $ \Ot (\sqrt{n}) $ rounds.

%% file: conclusion.tex
\section{Conclusion}

As explained in this paper, the algorithm of Lee and Sidford is based on an interior-point method that (1) performs $ \Ot (\sqrt{n}) $ iterations and (2) spends $ \Ot(m) $ operations per iteration using primitives like matrix-vector multiplication and solving a Laplacian system.
Recent advances show how to (slightly) improve upon the $ \Ot (m \sqrt{n}) $ bound by employing sophisticated dynamic data structures to decrease the (amortized) number of operations spent per iteration~\cite{BrandLLSS0W21,AxiotisMV21}.
In an even more recent breakthrough~\cite{maxmaxflow}, the minimum cost flow problem has been solved in $O(m^{1+o(1)})$ time in the centralized model. This is done by an interior point method that (1) performs $O(m^{1+o(1)})$ iterations and (2) spends $O(m^{o(1)})$ time per operation. 
One might wonder whether alternatively it would also be possible to improve upon the $ \Ot (m \sqrt{n}) $ bound purely by finding an interior-point method with a reduced number of iterations.

In this paper we have demonstrated how the primitives employed in each iteration can be carried out efficiently in the Broadcast Congested Clique and thus provide evidence that improvements in the iteration count of the interior-point method would likely carry over to the round complexity in the Broadcast Congested Clique -- currently this leads to an algorithm with $ \Ot (\sqrt{n}) $ rounds.
This motivates the question of whether a lower bound for the min-cost flow problem can be obtained in this model since such a lower bound would rule out the possibility of an improvement of the iteration count using the same fast primitives.
Note that in contrast to the (Unicast) Congested Clique -- for which lower bounds would yield a breakthrough in circuit complexity~\cite{DruckerKO13} -- the Broadcast Congested Clique seems much more amenable to lower bounds; in particular polynomial lower bounds exist for several graph problems in the broadcast model~\cite{FrischknechtHW12,DruckerKO13,Censor-HillelKK19,HolzerP15,BeckerMRT20}.